\documentclass[runningheads]{llncs}
\usepackage[T1]{fontenc}
\usepackage{graphicx}
\usepackage{multirow}%
\usepackage{amsmath}
\usepackage{xcolor}%

\newtheorem{assumption}{Assumption}
\newtheorem{observation}{Observation}
\begin{document}
\title{Games in Public Announcement: How to Reduce System Losses in Optimistic Blockchain Mechanisms}
%
\titlerunning{How to Reduce System Losses in Optimistic Blockchain Mechanisms}
%
\author{Siyuan Liu\inst{1} \and Yulong Zeng\inst{2}}
%
%
\institute{School of Software and Microelectronics, Peking University, Beijing, China \and Beijing YeeZTech Ltd,  Beijing, China\\}
%
\maketitle              
\begin{abstract}
Announcement games, where information is disseminated by announcers and challenged by validators, are prevalent in real-world scenarios. Validators take effort to verify the validity of the announcements, gaining rewards for successfully challenging invalid ones, while receiving nothing for valid ones. Optimistic Rollup, a Layer 2 blockchain scaling solution, exemplifies such games, offering significant improvements in transaction throughput and cost efficiency. We present a game-theoretic model of announcement games to analyze the potential behaviors of announcers and validators. We identify all Nash equilibria and study the corresponding system losses for different Nash equilibria. Additionally, we analyze the impact of various system parameters on system loss under the Nash equilibrium. Finally, we provide suggestions for mechanism optimization to reduce system losses.

\end{abstract}

\section{Introduction}
Announcement games are prevalent in various real-world scenarios. In these games, announcers disseminate essential information on public platforms for specific purposes, such as PhD defenses or tender bidding. The announcement remains effective for a predetermined period, during which any party can challenge its validity. Upon raising an objection, the announcer and the objector engage in a contest stage, where the validity of both the objection and the announcement is scrutinized under supervision. The outcome determines the invalidity of either the announcement or the objection, with the loser incurring penalties and the winner gaining rewards. Importantly, if no valid objection is made during the announcement period, the announcement is deemed approved, even if it is invalid, thereby allowing the announcer to secure a significant advantage and and causing potential system losses.

A concrete example of an announcement game is the optimistic mechanism in the blockchain Layer-2(L2) ecosystem. Layer-1(L1) refers to the base layer of the blockchain architecture, such as Bitcoin or Ethereum's mainnet. L2 solutions, on the other hand, are secondary frameworks or protocols built on top of the L1 blockchain. They aim to enhance the scalability and efficiency of the blockchain without compromising its underlying security. 

L2 solutions achieve this by processing transactions\footnote{A transaction is an instruction to change the state of a blockchain, which can be initiated by anyone.} off the main blockchain and only recording final state on the L1 blockchain. A prominent L2 scaling solution is Optimistic Rollups. This approach allows certain users (known as aggregator) package a certain number of L2 transactions into a block and publish it on the L1, including a  final state indicating the processing result of packaged transactions. The validity of the state cannot be directly verified on L1 due to scalability limitations. However, any user is allowed to challenge a fraud state and provide necessary evidence, before the block is finalized. L1 blockchain is able to determine whether the evidence is correct. If so, the block is discarded and the L1 state reverts. If there is no valid evidence within the announcement period, the block is deemed finalized. Therefore, the optimistic mechanism, with its challenge-based verification process, provides a scalable and efficient solution for ensuring transaction validity, while still remaining a security issue that fraud blocks may be finalized.

Generally speaking, an optimistic rollup game includes the following roles:
\begin{itemize}
    \item \textbf{Aggregator}: The party proposing an L2 block, which can be either valid or invalid, corresponding to honest or attack actions, respectively. If an attack is chosen, the aggregator can specify an unlawfully but public earned income within the L2 block. In this case it is referred to as the attacker.
    \item \textbf{Validator}: The party responsible for verifying the validity of the L2 block proposed by the aggregator. The action space of the validator includes:
    \begin{itemize}
        \item \textbf{Honest-verifier} (verify): The validator actively verifies the validity of the L2 block and issues a challenge if the result of the verification is negative. In either case, the validator needs to bear the verification cost.
        \item \textbf{Free-rider} (no verify and no challenge): The validator neither verifies nor challenges, assuming the L2 block is correct, and has no cost.
        \item \textbf{Chance-taker} (no verify but challenge): The validator does not verify but issues a challenge, assuming the L2 block is incorrect, and has no cost.
    \end{itemize}
    If the validator chooses to verify, it is referred to as the verifier; if the validator chooses to challenge, it is referred to as the challenger. In L2 scenario, both the validator and the aggregator are required to stake a deposit, for penalizing dishonst actions.
\end{itemize}

Announcement games have broad applications, encompassing various domains such as public bidding, tender processes, academic announcements, and blockchain systems like Optimistic Rollup. A common issue in these scenarios is that when an incorrect announcement is finalized, the system incurs substantial losses. For instance, in Optimistic Rollup, if a malicious block is finalized, the attacker gains benefits that are undeserved, undermining the system's stability and credibility. Therefore, achieving lower system losses of announcement games is a crucial research focus, enhancing efficiency and reliability across these diverse fields. Our model addresses the complex behaviors and decision-making processes of participants. The probability of validator verification and the probability of aggregator attacking affect each other, leading to the formation of a Nash equilibrium. The situation becomes more intricate with multiple validators, as each validator's behavior is also affected by others. Based on the above findings, it is of interest to study the system losses in equilibrium and how to reduce these losses by adjusting system parameters and fine-tuning mechanisms.

In summary, the contributions of this paper include:
\begin{enumerate}
\item Modeling the Optimistic Rollup game and solving all Nash equilibria, including cases with a single validator and multiple validators. Moreover, we find that the number of equilibria is simply determined by the value of a certain system parameter.
\item Proposing additional mechanisms based on adjustments to the original optimistic mechanism, including the non-KYC scenario and breaking-tie strategies, aimed at mitigating system losses.
\item Analyzing the impact of different parameters on system losses and providing suggestions to optimize the design and efficiency of the system.
\end{enumerate}

\section{Related Works}
Blockchain technology's rapid growth has highlighted significant scalability issues in L1 blockchains, such as Bitcoin~\cite{blockchain} and Ethereum~\cite{ethereum}, which limit transaction throughput and speed due to their consensus mechanisms~\cite{scale}. L2 scaling solutions, such as State Channels, Plasma, Sidechains, and Rollups, address these limitations by operating on L2 chains to enhance scalability and efficiency~\cite{Louis}. Over the past year, the volume of atomic arbitrage MEV (Maximal Extractable Value) transactions on major L2 networks has exceeded \$3.6 billion, accounting for 1\% to 6\% of all DEX (decentralized exchange) trading volumes~\cite{sui}. Among these, Optimistic Rollups have gained substantial traction, with Rollups currently handling a significant portion of Ethereum's transaction volume~\cite{Thibault}. Optimistic Rollups, such as Arbitrum~\cite{nitro} and Optimism~\cite{optimism}, offer advantages like greater decentralization and compatibility with existing smart contracts by assuming transaction validity unless challenged. However, most current L2 solutions still face centralization issues. To address this, our article provides an in-depth analysis of the Optimistic Rollup system, offering essential theoretical insights for achieving full decentralization in the future.

Related studies in game theory have extensively explored announcement games in various contexts, highlighting the strategic behavior of participants in scenarios involving public announcements~\cite{cui}. Ågotnes et al.~\cite{van} analyze the rational strategies in public announcement games, combining logic and game theory in the study of rational information exchange. Loi Luu et al.~\cite{Loi} analyse the incentivization of validators in blockchain settings. Their study shows that practical attacks exist which either waste miners’ computational resources or lead miners to accept incorrect script results, known as the verifier’s dilemma. These studies also highlight that announcement games are a practical scenario in blockchain. Hans et al.~\cite{Hans} study blockchain security through the lens of game theory, focusing on the design of reward-sharing mechanisms for validation. Another related paper is~\cite{Lars}, where costs of validation differ, and the authors look at the problem of delegating validation.

A concurrent study by Li~\cite{Jiasun} engages with equilibrium in optimistic rollup. The author presents a model of a potential attack on the well functioning of optimistic rollups and concludes that their current design is not secure. However, the analysis does not take the deposit of validators into consideration and erroneously assumes uniform behavior among validators, resulting in flawed conclusions. Similarly, another concurrent study~\cite{Akaki} also investigates equilibrium in optimistic rollup, providing both lower and upper bounds on the optimal number of validators, and advise on optimal design of rewards for optimal design of rewards. However, their model is criticized for its simplicity and lack of consideration for system benefits. Daji et al.~\cite{Daji} first focus on the behavior of chance-taker validators. Their study analysis the equilibrium of a single aggregator and validator, although it does not extend to multi-player scenarios. Our model is optimized based on the aforementioned works. We include the chance-taker as an optional behavior for validators and assume that each validator's behavior is independent. Additionally, we define the benefits of various behaviors for different roles and introduce system rewards. Finally, our model also considers scenarios involving multiple validators.

\section{Model}
Our model mainly consists of two types of roles: the aggregators $\mathcal{A}$ and the validators $\mathcal{V}$. Aggregators $\mathcal{A}$ propose L2 blocks, which can be either valid or invalid. To deter fraudulent proposals, aggregators must stake a deposit, which they forfeit if their block is invalidated. They aim to maximize their earnings by including as many transactions as possible in the proposed blocks. Validators $\mathcal{V}$ verify the validity of the L2 blocks. They have three strategic choices: verifier, free-rider or chance-taker. Validators also act to maximize their utilities. We state some explanations and assumption in this section. We introduce the concept of Know Your Customer(KYC). In the KYC scenario, validators stake a fixed deposit, while in the non-KYC seenario, they can choose their deposit amount.

\subsection{Parameters}
Table~\ref{table:notation} presents the various parameters of the game.
\begin{table}[]
    \caption{The Notation of Parameters}
    \begin{tabular}{|c|l|}
    \hline
    Notations & The Definition of Notations\\
    \hline
    Z & Malicious Block Value \\ 
    S & Aggregator's deposit \\ 
    B & Aggregator reward \\ 
    T & Validator reward \\ 
    n & The number of validators \\ 
    $V$ & the validators's deposit.  \\ 
    C & Cost of validation \\ 
    $\delta$ & The proportion of the attacker's deposit that the challenger receives.\\ 
    $f_p$ & The proportion of penalty deposit when the validator pledges the wrong block.\\ 
    $f_n$ & The proportion of penalty deposit when the validator challenges the right block.\\
    \hline
    \end{tabular}
    \label{table:notation}
    \footnotetext{Note: The cost of the aggregator packaging transactions into the block can be normalized.}
\end{table}

\begin{itemize}
    \item {\itshape Malicious block value $Z$}: The attacker can arbitrarily assign a value to the malicious block, which he gains if that block is finalized. This value includes the aggregator's reward.

    \item {\itshape Deposit $S$ and $V$}: When dishonest players are caught, their deposits are forfeited. The system requires each aggregator to stake a fixed deposit $S$ to penalize malicious behavior. The same applies to validators. In the KYC scenario, we require the verifier to pledge a fixed deposit $V$ on the block. We also introduce the non-KYC scenario, where the validator can arbitrarily choose their deposit\footnote{If a validator wants to leverage his deposit $V$, up to a maximum amount of founds the validator can raise, he can create multiple accounts, which is free in the blockchain system, each staking the required deposit amount. It is known as the Sybil attack. As a result, the validator's payoff is also leveraged while the validation cost remains singular. This is equivalent to validators being able to choose their deposit amounts, with their payoffs directly proportional to their deposit amounts.  }, up to a maximum value $V_{max}$.

    \item {\itshape System reward $B$ and $T$}: If a block is finalized, the system allocates two constant amounts for rewarding the aggregator and the validator set, denoted by the \emph{aggregator reward} $B$ and the \emph{validator reward} $T$. Each validator then receives a portion of the validator reward proportional to their deposit amount. This ensures that system payouts do not exceed a certain amount.

    \item {\itshape Cost $C$}: We assume that the cost for each honest validator is constant, referring to the computational cost of validation. There is no cost for dishonest (free-rider and chance-taker) validators. Additionally, the cost for aggregators is normalized to zero.

    \item {\itshape Penalty proportion of aggregator $\delta$}: When an aggregator behaves maliciously and is challenged by a validator, the aggregator forfeits their deposit $S$, with a portion $\delta S$ awarded to the validator. This ensures that payouts come solely from malicious actors, deterring collusion between aggregators and validators to exploit the system.

    \item {\itshape Proportion of false positive $f_p$ and false negative $f_n$}: If a validator incorrectly pledges a block, a penalty of $f_pV$ is imposed. Similarly, if a validator wrongly challenges a correct block, resulting in a failed challenge, a penalty of $f_nV$ is imposed and awarded to the aggregator.

    By definition, $\delta$, $f_p$, and $f_n$ are all within the range of $[0,1]$.
\end{itemize}

We make some assumptions on these parameters to assure that there is no dominate behavior for each party.

\begin{assumption}
    The validator does not play a dominated strategy, that is $\delta S > T$.
\end{assumption}

We hope validators to challenge incorrect blocks instead of colluding with aggregators, so the reward for finding incorrect blocks should be greater than the reward for confirming blocks, that is $\delta S > T$. Additionally, we do not require $T > C$, for the reason that even the validator has a negative benefit of verifying, he may still choose to verify because once the wrong block is discovered, there will be a positive benefit.

\begin{assumption}
    The aggregator does not play a dominated strategy, that is $Z>B$.
\end{assumption}

For the aggregator, choosing a larger value $Z$ does not incur additional cost, as the difference in block values is only a numerical difference throughout the entire block. Moreover, the block value is a public information because it is on chain. So the validators can choose their behavior according to the block value.

\subsection{Payoff Matrix}
Based on the above analysis, when there is only one aggregator and one validator, the validator will not incur any penalties when choosing the free-rider strategy.

In most scenarios, the chance-taker strategy is dominated because validators must verify to win. In L2 scenarios, precise verification of state transitions is required. Without it, validators can't win challenges even if the block is incorrect, due to the burden of proof. However, in cases where official authorities handle validation, such as reporting cheating in a game, the chance-taker strategy is viable. This paper focuses on the chance-taker case in a single validator game.

Table \ref{table:matrix11} presents a game involving one aggregator and one validator in a bimatrix format.

\begin{itemize}
    \item If the aggregator attacks and the validator challenges, the aggregator loses his deposit and the validator earns a portion of the aggregator's deposit. Additionally, the validator incurs costs if he verifies.
    
    \item If the aggregator attacks and there is no challenger, the aggregator earns the malicious block value, and the validator receives the validator reward.
    
    \item A honest aggregator always earns the aggregator reward. The validator receives the validator reward if he does not challenge. Honest validators incur costs, while free-riders pay nothing. However, if the validator challenges the correct block, he loses a portion of his deposit, which goes to the aggregator.
\end{itemize}

\begin{minipage}{\textwidth}
\begin{minipage}[t]{0.48\textwidth}
\makeatletter\def\@captype{table}
\scalebox{0.8}{
\renewcommand\arraystretch{1.2}
\begin{tabular}{|c|c|c|}
    \hline
    & Not Attack  & Attack  \\ 
    \hline
    Free-rider  & (B, T) & (Z, T) \\ 
    Chance-taker & ($B+f_n V , -f_n V$) & ( $ -S, \delta S $) \\ 
    Verifier & (B, T-C) & ($-S, \delta S -C$) \\ 
    \hline
\end{tabular}}
\caption{Payoff matrix of one aggregator and one validator. The first number indicates the utility for the aggregator, while the second number represents the utility for the validator.}
\label{table:matrix11}
\end{minipage}
\begin{minipage}[t]{0.48\textwidth}
\makeatletter\def\@captype{table}
\scalebox{0.8}{
\renewcommand\arraystretch{1.2}
\begin{tabular}{|c|c|c|c|}
    \hline
    & Not Attack & \multicolumn{2}{c|}{Attack} \\ 
    \hline
    \multirow{2}{*}{Free-rider} & \multirow{2}{*}{$(B, \frac{T}{n})$} & Detected & $(-S, -f_pV)$ \\
    & & Not Detected & $(Z, \frac{T}{n})$ \\
    \hline
    Verifier & $(B, \frac{T}{n}-C)$ & \multicolumn{2}{c|}{$(-S, \frac{\delta S}{m} -C)$} \\
    \hline
\end{tabular}}
\caption{Payoff matrix of one aggregator and $n$ validators, where $m$ of the $n$ validators choose to verify. The first number indicates the utility for the aggregator, while the second number represents the utility for the validator.}
\label{table:matrix1n}
\end{minipage}
\end{minipage}

\subsection{Multiple Players}
We assume that when there are multiple challengers, they equally share the challengers' benefit $\delta S$, as only the first challenger will gain and each challenger has equal probability to be the first one. 

Considering multiple aggregators is pointless because there is no strategic interaction between
aggregators. Validators join the game based on the first
proposed block. If the first proposed block is invalid, the game proceeds with the
second proposed block, and so forth. 

Table~\ref{table:matrix1n} presents a game involving one aggregator and \( n \) validators in a bimatrix format. The first number represents the utility (payoff) for the aggregator, and the second number represents the utility for the validator. We assume that \( m \) out of \( n \) validators choose to verify. We also do not consider the chance-taker case with multiple players.

\begin{itemize}
    \item If the aggregator attacks and is detected by any validator, the aggregator loses his deposit, and each challenger's earnings become part of the aggregator's deposit, distributed evenly among the challengers in expectation.
    
    \item If the aggregator attacks and there is no challenger, the aggregator earns the malicious block value, and validators share the validator reward proportionally to their deposits.
    
    \item If the aggregator is honest, he always earns the aggregator reward, and validators share the validator reward proportionally to their deposits. Honest validators incur costs, while free-riders pay nothing.
\end{itemize}

\section{Equilibrium Analysis}
In this section, we analyze the game in KYC scenario where the validator must deposit a fixed amount of \( V \) dollars on the block. 

\subsection{One validator}
We denote \( \beta \) as the probability of \( \mathcal{A} \) attacking and \( \alpha \) as the probability of \( \mathcal{V} \) verifying. Upon not verifying, \( \mathcal{V} \) challenge with probability \( \gamma \) (so w.p. $1-\alpha-\gamma$, ). Recall the payoff matrix of Table~\ref{table:matrix11} from Section 3.


\begin{lemma} \label{lemma1}
There is no pure strategy equilibrium between $\mathcal A$ and $\mathcal V$.
\end{lemma}

Due to space limitations, the proofs of the lemmas and theorems in this paper are included in the appendix.

Since there is no pure-strategy equilibrium in the game, we consider the mixed strategy equilibrium. For the aggregator \( \mathcal{A} \), the expected return of attacking should be equal to the expected return of choosing not to attack. For the validator \( \mathcal{V} \), considering that the size of its action space is 3, it can arbitrarily choose two of the behaviors as a mixed strategy, and the expected return of each behavior should also be equal.

\begin{theorem} \label{theorem1}
    There is a Nash Equilibrium that:
    \begin{itemize}
        \item If $\displaystyle C > \frac{(\delta S-T)(T+f_n V)}{\delta S+f_n V}$, $\mathcal A$ attacks with probability $\displaystyle \beta = \frac{T+f_n V}{\delta S+f_n V}$, while $\mathcal V$ challenges with probability $\displaystyle \gamma = \frac{Z-B}{Z+S+\lambda f_n V}$; 

        \item If $\displaystyle C < \frac{(\delta S-T)(T+f_n V)}{\delta S+f_n V}$, $\mathcal A$ attacks with probability $\displaystyle \beta = \frac{C}{\delta S-T}$, while $\mathcal V$ verifies with probability $\displaystyle \alpha = \frac{Z-B}{Z+S}$.

        \item Especially, when $\displaystyle C=\frac{(\delta S-T)(T+f_n V)}{\delta S+f_n V}$, there is an additional equilibrium that $\mathcal A$ can choose to attack with probability $\displaystyle \beta = \frac{T+f_n V}{\delta S+f_n V}$, while $\mathcal V$ can choose whether to verify with probability $\displaystyle \alpha = \frac{Z-B}{Z+S}$ and whether to challenge with probability $\displaystyle \gamma = \frac{Z-B}{Z+S+\lambda f_n V}$ if not verifying. Combining the first two equilibrium, there are three equilibriums in total.
    \end{itemize}
    
\end{theorem}



\begin{proof}[Sketch]
    In equilibrium, the aggregator and the validator are indifferent between their each behaviors. We obtain the result of Theorem~\ref{theorem1} by calculating the utilities corresponding to all possible behaviors of the aggregator and the validator. See the Appendix for the specific proof.
\end{proof}

We can learn from the equilibrium above that system loss mainly comes from when a malicious aggregator proposes an incorrect block that has not been verified by the verifier. The system loss of one aggregator and one validator is, 
\begin{equation}
    \mathcal{L}=\beta(1-\alpha)Z=\frac{C}{\delta S-T} \frac{(S+B)Z}{S+Z}
    \label{loss}
\end{equation}

Consider Eq.(\ref{loss}) as a function related to $Z$ and take its derivative, that is, 
    \begin{equation}
        \mathcal{L}'(Z) = \frac{C(S+B)}{\delta S-T} \frac{S}{(S+Z)^2}
    \label{loss'}
    \end{equation}

From Eq.(\ref{loss'}), we can learn that $\mathcal{L}'(Z) > 0$ always holds. So as $Z$ increases, $\mathcal{L}$ increases. 

\subsection{Extension to two validators}

When there are two validators, the payoff logic changes depending on the action of the other validator. To simplify the model, we do not consider the case of chance-takers when there are two or more validators. We denote the two validators as \( \mathcal{V}_1 \) and \( \mathcal{V}_2 \), with the probabilities of verifying being \( \alpha_1 \) and \( \alpha_2 \), respectively.

\begin{theorem} \label{theorem2}
    Denote $\displaystyle R=\frac{\frac{T}{2}+f_pV}{\delta S}$. There is an equilibrium that both $\mathcal V_1$ and $\mathcal V_2$ play the mixed strategy that they verify with probability $\displaystyle \alpha = 1-\sqrt{\frac{B+S}{Z+S}}$, while $\mathcal A$ attacks with probability $\displaystyle \beta_1 = \frac{C}{\delta S (1-\frac{1}{2} \alpha) + \alpha(\frac{1}{2}T+f_pV) - \frac{1}{2}T}$.
    
    As $\displaystyle R \le \frac{1}{2}$, in addition to the above equilibrium, there is another equilibrium that $\mathcal V_1$ plays the mixed strategy verifying with probability $\displaystyle \alpha = \frac{Z-B}{Z+S}$, and $\mathcal V_2$ plays as a free-rider, while $\mathcal A$ attacks with probability $\displaystyle \beta_2 = \frac{C}{\delta S - \frac{T}{2}}$.
\end{theorem}

\begin{proof}[Sketch]
    We first discuss that the case where both validators play pure strategies does not exist. Then we obtain the equilibria of the other cases similar to Theorem~\ref{theorem1}. See the Appendix for the specific proof.
\end{proof}

The system loss of the equilibrium where both validators play the mixed strategy, denoted as \( \mathcal{L}_1 \), is:
\begin{equation*}
    \mathcal{L}_1 = \beta_1 \prod_{i=1}^2 (1 - \alpha_i) Z = \frac{C}{\delta S \left(1 - \frac{1}{2} \alpha \right) + \alpha \left(\frac{1}{2} T + f_p V \right) - \frac{1}{2} T} \frac{(S + B) Z}{S + Z}.
\end{equation*}

Similarly, the system loss of the equilibrium where one validator plays the mixed strategy while the other validator free-rides, denoted as \( \mathcal{L}_2 \), is:
\begin{equation*}
    \mathcal{L}_2 = \beta_2 \prod_{i=1}^2 (1 - \alpha_i) Z = \frac{C}{\delta S - \frac{1}{2} T} \frac{(S + B) Z}{S + Z}.
\end{equation*}

Since \( \beta_1 > \beta_2 \), we can deduce that \( \mathcal{L}_1 > \mathcal{L}_2 \). Therefore, the system tends to set \( R \le \frac{1}{2} \) to minimize system loss.

\subsection{Extension to $n$ validators}
Based on the above analysis, we extend this model to multiplayer games. As mentioned earlier, only one correct block will eventually be finalized, so we still consider only one aggregator and scale the validator size to \( n \).

\begin{proposition}
    In equilibrium, the aggregator does not play a pure strategy.
    \label{propos1}
\end{proposition}

If the aggregator always attacks, validators will challenge the malicious blocks, causing the aggregator to lose their deposit and switch to honest behavior. Conversely, if the aggregator never attacks, validators will free-ride, prompting the aggregator to start attacking for higher revenue.

\begin{proposition}
    In equilibrium, with $n$ validators, there is no strategy where all validators are free-riders or one validator purely verifies.
    \label{propos2}
\end{proposition}

Assuming that all validators are free-riders, the aggregator has a strong incentive to attack. In this scenario, if any validator switches their behavior from a free-rider to a verifier, they will achieve greater utility. Consequently, this situation cannot constitute an equilibrium.


Building on Proposition~\ref{propos1} and Proposition~\ref{propos2}, we explore mixed strategy behaviors. \textbf{Suppose that there are \( n \) validators, and \( m \) of these \( n \) validators play the mixed strategy "verifier or free-rider" while the remaining \( n - m \) parties play the strategy "free-rider". For the \( i \)-th party of the \( m \) mixed strategy parties, the probability of verifying is \( \alpha_i \). Denote \( \beta \) as the probability of the aggregator attacking.}

\begin{lemma}\label{lemma2}
    In equilibrium with one aggregator and $n$ validators, the behavior of the validator group adheres to the following rules:
    \begin{enumerate}
        \item $k$ validators adopt a mixed strategy with probability $\alpha_1$;
        \item $m-k$ validators adopt a mixed strategy with probability $\alpha_2$;
        \item $n-m$ validators play a pure strategy as free-riders.
    \end{enumerate}
    A special case arises when $\alpha_1 = \alpha_2$, indicating that all validators employing a mixed strategy are symmetric.
\end{lemma}

According to Lemma \ref{lemma2}, validators choosing a mixed strategy exhibit at most two distinct verification probabilities. We define a single mixed strategy probability as the symmetric case and two different probabilities as the asymmetric case. The following sections will analyze these cases separately.

\subsubsection{The Symmetric Case}
Before analyzing the equilibrium in the symmetric case, it is important to introduce the following definition to provide a clear understanding of the terms used in our analysis.
\begin{definition}
$m$-NE Definition: $m$ is termed an ``m-Nash Equilibrium" ($m$-NE) if there is a Nash Equilibrium with $m$ validators using mixed strategies and $n - m$ validators being free-riders.
\end{definition}


\begin{definition}
    We define the following variables:
    \begin{subequations}
        \begin{align*}
            & A = \frac{B+S}{Z+S}; \\
            & \displaystyle R = \frac{\frac{T}{n}+f_pV}{\delta S};\\
            & \displaystyle P_m = \frac{1-A}{m\alpha_m}, \displaystyle Q_m = \frac{A}{1-\alpha_m};\\
            & \Delta_m = \frac{P_m-P_{m+1}}{Q_m-Q_{m+1}};\\
            & \Gamma_m = \left[\frac{1}{m(m+1)}\left(\frac{1}{A}-1\right)-\frac{\alpha_m}{m+1}\right]\frac{1-\alpha_m}{\alpha_m^2}(m>0).\\
        \end{align*}
    \end{subequations}  
\end{definition}

It is not difficult to prove that both $P_m$ and $Q_m$ are within the range of $(0,1)$ and both are increasing functions. Especially, when m approaches zero, $\Gamma_m$ also approaches zero, so we stipulate $\Gamma_0 = 0$. We now proceed to state the equilibrium for the case where there is one aggregator and \(n\) validators, with \(m\) in \(n\) validators playing the mixed strategy symmetrically.  We denote $\alpha_m$ as the probability of each verifier when there are $m$ mixed strategy validators.

\begin{lemma} \label{lemma3}
    In the equilibrium of the symmetric case, the following conditions must be met:
    \begin{subequations}
    \begin{align}
        & \beta (P_m \delta S - Q_m (f_pV+\frac{T}{n}) + f_pV) = C; \\
        & R \le \Gamma_m \label{r<gamma}\ (if\  m < n);\\
        & A=\frac{B+S}{Z+S}=(1-\alpha_m)^m.
    \end{align}
    \end{subequations}
\end{lemma}

\begin{proposition}
    In m-NE, the probability of the aggregator attacking is $\displaystyle \beta = \frac{C}{P_m \delta S - Q_m (f_pV+\frac{T}{n}) + f_pV}$, and the probability of each mixed strategy validator is $\displaystyle \alpha_m = 1- \sqrt[m]{\frac{B+S}{Z+S}}$.
    \label{proposition3}
\end{proposition}


\begin{lemma} \label{lemma_Gamma} \label{lemma4}
    As \(m\) increases, \(\Gamma_m\) increases.
\end{lemma}

\begin{lemma} \label{lemma5}
    As \(m\) increases, \(\Delta_m\) increases.
\end{lemma}

\begin{lemma} \label{lemma6}
    For \(\forall m \le n\), \(\Gamma_m < \Delta_m\).
\end{lemma}

Based on the above lemmas, we can now proceed the main theoerm of the symmetric case as follows.
\begin{theorem}[Main Theorem] \label{main_theorem}
    $n$-NE always exists. When $\Gamma_{m-1} < R \le \Gamma_{m} \ (0<m<n)$, additional equilibriums are the following $n-m-1$ equilibrium: $m$-NE, $(m+1)$-NE, ..., $(n-1)$-NE. The probabilities in the equilibrium $\beta$ and $\alpha_m$, are as stated in Proposition~\ref{proposition3}.
    
    Among these $n-m$ equilibriums, the system loss increases as the number of verifiers increases, that is $\beta_m < \beta_{m+1} < \dots < \beta_n$, such that $\mathcal{L}_m < \mathcal{L}_{m+1} < \dots < \mathcal{L}_n$.
\end{theorem}

\begin{proof}[Sketch]
    We draw the first conclusion from Proposition~\ref{proposition3} and Lemma~\ref{lemma4}. 
    According to the definition of $\beta$ in Proposition~\ref{proposition3}, $\beta_m < \beta_{m+1}$ can be simplified to $R<\Delta_m$. Then we can derive the second conclusion of Theorem~\ref{main_theorem} is from Lemma~\ref{lemma5} and Lemma~\ref{lemma6}. 
\end{proof}

In summary, by leveraging the lemmas and definitions established earlier, we have shown that the conditions for the equilibrium in the symmetric case can be determined by the relationship between \(\Gamma_m\) and \(R\). And the optimal equilibrium when $\Gamma_{m-1} < R \le \Gamma_{m}$ is when there are $m$ mixed strategy validators. See the Appendix for the specific proof.

The maximum system loss for the symmetric case is when there are \(n\) mixed strategy validators, which is given by:
\begin{equation*}
    \mathcal{L}_{sym_n} = \beta_n \prod_{i=1}^n (1-\alpha_i) Z = \frac{C}{P_n \delta S - Q_n (f_pV + \frac{T}{n}) + f_pV} \cdot \frac{(S+B)Z}{S+Z}.
\end{equation*}

\subsubsection{The Asymmetric Case}
Next, we consider the asymmetric case. 

Suppose there are \(k\) validators who play the mixed strategy with probability \(\alpha_1\), and \(m-k\) validators who play the mixed strategy with probability \(\alpha_2\). This scenario introduces variability in validator strategies, making the analysis more complex than the symmetric case. To facilitate the following analysis, we assume \(\alpha_1 < \alpha_2\).

Let $\displaystyle p_3 = \sum_{i=0}^{k-1}C_{k-1}^i \alpha_1 ^i (1-\alpha_1)^{k-1-i}  \sum_{j=0}^{m-k-1}C_{m-k-1}^j \alpha_2 ^j (1-\alpha_2)^{m-k-1-j} \frac{1}{i+j+2}$, 
$\displaystyle p_4 = \sum_{i=0}^{k-1}C_{k-1}^i \alpha_1 ^i (1-\alpha_1)^{k-1-i}  \sum_{j=0}^{m-k-1}C_{m-k-1}^j \alpha_2 ^j (1-\alpha_2)^{m-k-1-j} \frac{1}{i+j+1}$, \(p_5 = (1-\alpha_1)^{k-1} (1-\alpha_2)^{m-k-1}\), we state the lemma as follows.

\begin{lemma} \label{lemma7}
    The necessary conditions for the equilibrium of asymmetric case are:
    \begin{subequations}
    \begin{align}
        & (p_3 - p_4) \delta S + p_5 (f_p V + \frac{T}{n}) = 0 \label{eqp1} \\
        & \frac{C}{\beta} - p_4 \delta S - f_p V + p_5 (f_p V + \frac{T}{n}) = 0 \label{eqp2} \\
        & (1-\alpha_1)^k (1-\alpha_2)^{m-k} Z - \left[1 - (1-\alpha_1)^k (1-\alpha_2)^{m-k}\right] S = B \label{f23}
    \end{align}
    \end{subequations}
\end{lemma}

Combining Eq.~(\ref{eqp1}), Eq.~(\ref{eqp2}), and Eq.~(\ref{f23}), we can determine the trend of \(\alpha_1\), \(\alpha_2\), and \(\beta\) as \(R\) changes. We set \(n=15\), \(m=10\), and other necessary constants. We calculated the probabilities \(\alpha_1\) and \(\alpha_2\) as shown in Fig.~\ref{fig:m11}. 

\begin{figure}
\includegraphics[width=\textwidth]{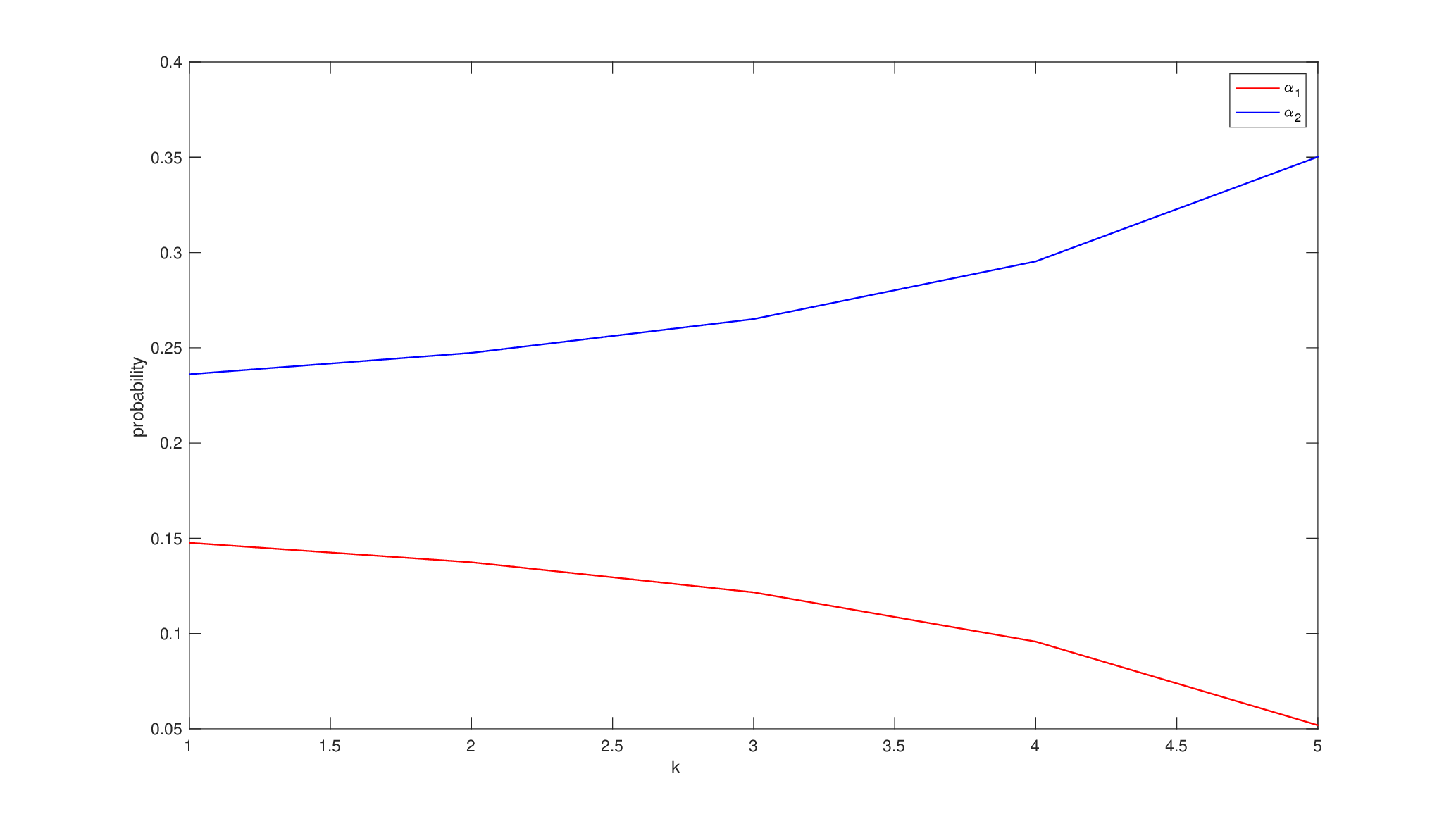}
 \caption{The relationship between the probabilities \textcolor{red}{$\alpha_1$} and \textcolor{blue}{$\alpha_2$} as the number of verifier \(k\) changes.}
 \label{fig:m11}
\end{figure}

\begin{observation}
    In equilibrium of the asymmetric case, the probability of verification will be greater for the group with more mixed strategy validators, i.e. $k \leq m/2$.
\end{observation}

From Eq.~(\ref{eqp2}), we can derive that:
\begin{equation*}
\begin{aligned}
    \beta = \frac{C}{p_4 \delta S  - p_5 (f_p V + \frac{T}{n}) + f_p V}.
\end{aligned}
\end{equation*}

Fig.~\ref{fig:beta20010v1} shows the variation of \(\beta\) values corresponding to different \(k\) values as \(R\) changes when \(m=10\). As \(k\) increases, the value of \(\beta\) decreases, resulting in a reduction in system losses. As \(k\) gradually approaches \(\frac{m}{2}\), an additional \(\beta\) value appears, and this \(\beta\) will gradually converge with the other \(\beta\). Similarly, we set \(m\) to 11 while keeping other variables unchanged, resulting in Fig.~\ref{fig:beta20011v}.

The system loss for the asymmetric case is given by:
\begin{equation}
    \mathcal{L}_{asym} = \beta_k \prod_{i=1}^n (1-\alpha_i) Z = \frac{C}{p_4 \delta S  - p_5 (f_p V + \frac{T}{n}) + f_p V} \cdot \frac{(S+B) Z}{S+Z}.
\end{equation}

From Figs.~\ref{fig:beta20010v1} and \ref{fig:beta20011v}, we observe that \(k=0\) represents the symmetric case where all mixed strategy validators play with the same probability, and \(\beta\) is maximized when \(k=0\). Since \(\beta_{sym} > \beta_{asym}\), we can infer that the system loss in the symmetric case \(\mathcal{L}_{sym}\) is greater than the system loss in the asymmetric case \(\mathcal{L}_{asym}\).

\begin{figure}
\includegraphics[width=\textwidth]{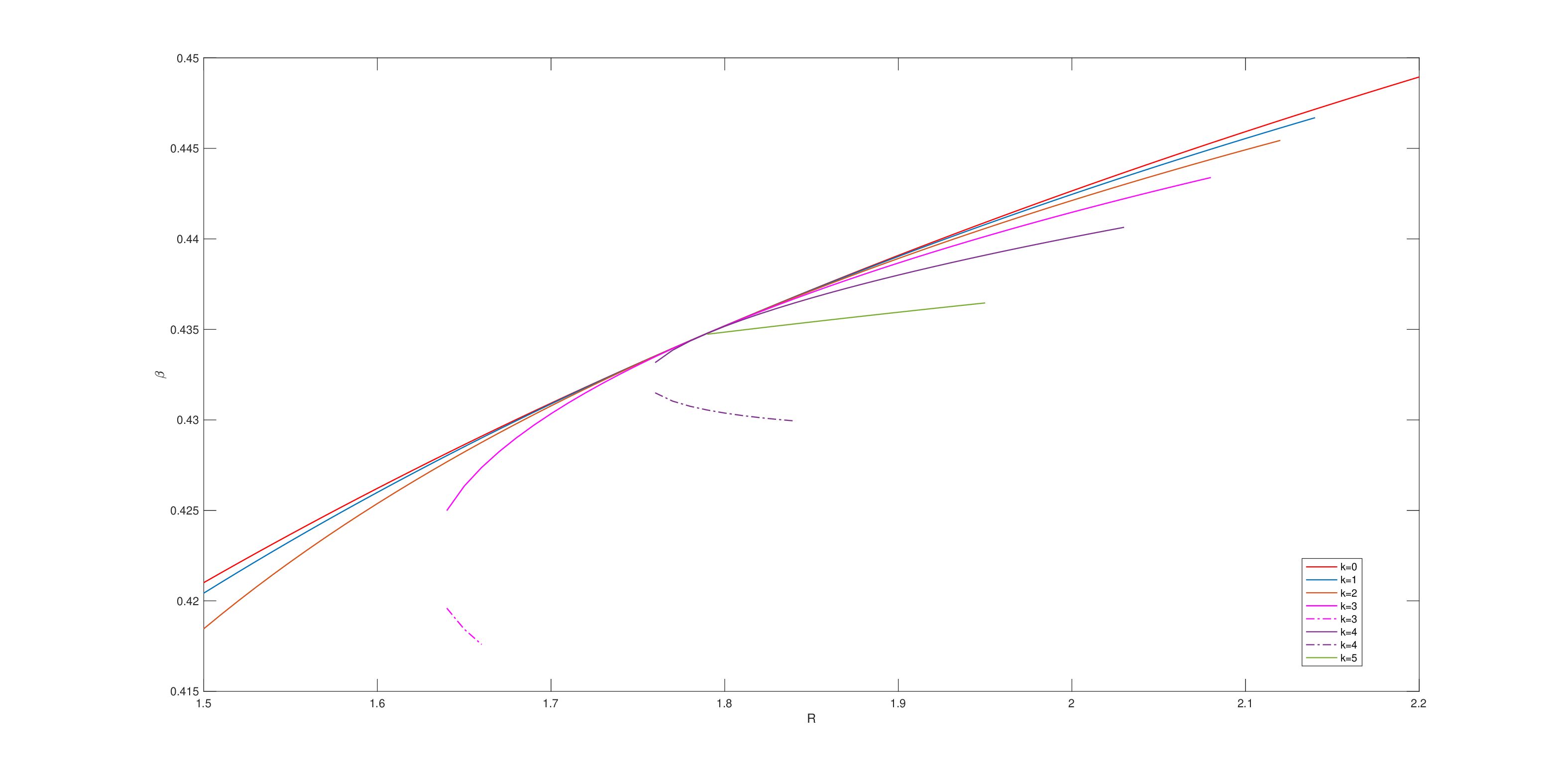}
\caption{The image of the $\beta$ as validators k and R changes. The solid line represents the general $\beta$ that exists for all k, while the dashed line represents the additional $\beta$ that appears as k changes.}\label{fig:beta20010v1} 
\end{figure}

\begin{figure}
\includegraphics[width=\textwidth]{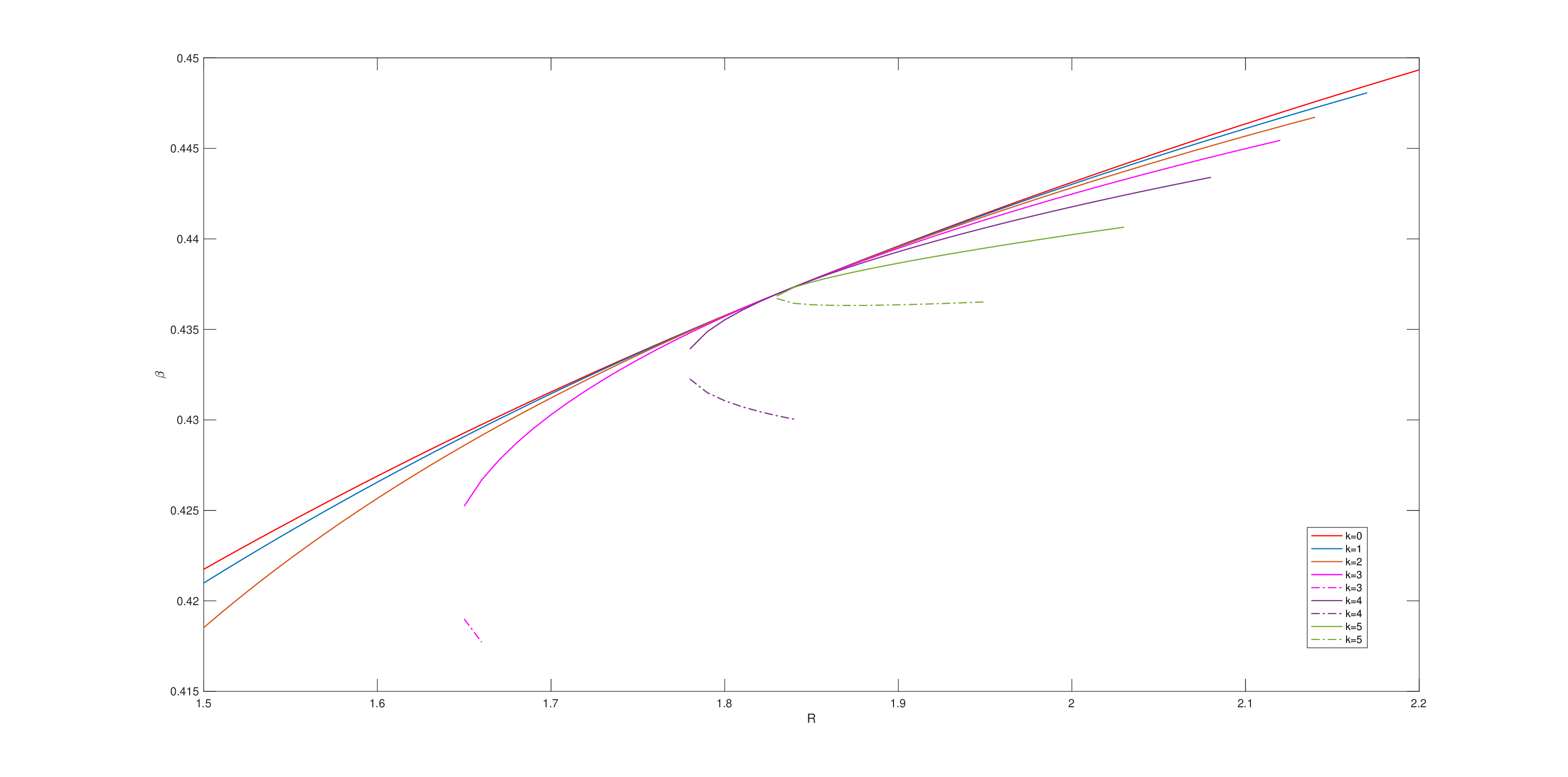}
\caption{The image of the $\beta$ as validators k and R changes. The solid line represents the general $\beta$ that exists for all k, while the dashed line represents the additional $\beta$ that appears as k changes.}\label{fig:beta20011v}
\end{figure}

\section{Non-KYC Scenario}
In the previous section, we discussed the game involving one aggregator with one validator, two validators, and multiple validators under the KYC model. In the KYC model, the deposit \( V \) required by the verifier to participate in the game needs to be a fixed value. Now, we extend the above discussion to non-KYC models, where the validator can deposit an arbitrary amount of stake. A common scenario under the non-KYC model is blockchain. In blockchain, validators can choose the deposit amount arbitrarily. 

When there is only one aggregator and one validator, it is the same as the case in the KYC scenario. Suppose there is one aggregator and two validators. Based on previous analysis, there is no pure strategy equilibrium. Consider the case where the aggregator plays the mixed strategy, attacking with probability \( \beta \), and the validators \( P_1 \) and \( P_2 \) play the strategy as verifier or free-rider with their deposits \( V_1 \) and \( V_2 \). Assuming both of their deposits are in the range of \( [V_{min}, V_{max}] \).

\begin{proposition}
    When a validator verifies, he always deposits the max amount $V_{max}$.
\end{proposition}

We consider the case when both validators play as the mixed strategy. Assume that $P_1$ plays the mixed strategy of verifier with a probability $\alpha_1$ and free-rider with a probability $1-\alpha_1$, while $P_2$ plays the mixed strategy of verifier with a probability $\alpha_2$ and free-rider with a probability $1-\alpha_2$. Denote deposit of $(P_1/P_2)$ is $(V_1^*/V_2^*)$. When $P_1/P_2$ plays as the verifier, they will deposit $V_{max}$. 

\begin{theorem} \label{theorem4}
    In equilibrium with two symmetric validators in non-KYC scenario, the validators either verify with a probability $\displaystyle \alpha_1 = \alpha_2 =1-\sqrt{\frac{B+S}{Z+S}}$ staking $V_{max}$ or not verify staking $V^*$. The system loss of non-KYC model $\mathcal{L}_{nonKYC}$ is less than $\mathcal{L}_{KYC}$ when the fixed deposit $V$ in the KYC scenario is selected as $V^*$.
\end{theorem}

The system loss is positively related to $\beta$ as, 
\begin{equation*}
\begin{aligned}
    \mathcal{L}=\beta\prod_{i=1}^n(1-\alpha_i)Z=\beta \frac{(S+B)Z}{S+Z}.
\end{aligned}
\end{equation*}

The aggregator's probability of attacking $\beta_{KYC}$ in the KYC scenario is, 
\begin{equation*}
\begin{aligned}
    \beta_{KYC} = \frac{C}{\delta S (1-\frac{1}{2} \alpha) + \alpha(\frac{1}{2}T+f_pV) - \frac{1}{2}T}.
\end{aligned}
\end{equation*}

The aggregator's probability of attacking $\beta_{KYC}$ in the non-KYC scenario is,
\begin{equation*}
\begin{aligned}
    \beta_{nonKYC}=\frac{C-\frac{V_{max}T}{V_{max}+V^*}+\frac{1}{2}T}{ \delta S(1-\frac{1}{2}\alpha)+\alpha(\frac{1}{2}T+f_pV^*)-\frac{V_{max}T}{V_m+V^*}}.
\end{aligned}
\end{equation*}

Since $\beta_{KYC} > \beta_{nonKYC}$ when the fixed deposit $V$ in the KYC scenario is selected as $V^*$ in nonKYC scenario, $\mathcal{L}_{KYC} > \mathcal{L}_{non-KYC}$.

\section{Breaking Tie}
The system loss primarily arises from the attack behavior of the aggregator. Suppose there are \( n \) validators, and the \( i \)-th validator has a probability \( \alpha_i \) of verifying.

\begin{equation}
    \mathcal{L} = \beta \prod_{i=1}^n (1 - \alpha_i) Z = \beta \frac{(S + B) Z}{S + Z}.
    \label{loss1}
\end{equation}

From Eq.(\ref{loss1}), we can see that \( \mathcal{L} \) increases as the value of \( Z \) increases. However, \( Z \) is uncontrollable by the system because \( Z \) is the block value chosen by the aggregator. In all before-mentioned cases, the aggregator's expected utility is indifferent with $Z$, which means that he does not have a preference to $Z$.

To achieve the purpose of breaking tie, we introduce a system interference term $D$ so that the penalty of attacking is related to the probability of the aggregator attacking $\beta$ and the validator verifying $\alpha$.

The new payoff matrix is shown in Table \ref{table:matrixd}. The reason why we only add interference term $D$ to the bottom right corner of the payoff matrix is that, only in this scenario the validator and aggregator have engaged in a battle on L1, which can be detected by the blockchain system. This mechanism can be implemented as follows: dynamically record the probability 
$p$ of the aggregator's attack being challenged. Each time this occurs, in addition to penalizing the aggregator, provide a return amount based on $p$ multiplied by a constant.

\begin{table}
    \centering
    \caption{Payoff matrix with a system interference term $D$}\label{table:matrixd}
    \begin{tabular}{|l|l|l|}
        \hline
        & $\overline A$ $(1 - \beta)$ & $A$ $(\beta)$ \\
        \hline
        $\overline {VC}(1-\alpha)$ & (B, $T$) & (Z, $T$) \\ 
        V $(\alpha)$ & (B, $T-C$) & $(-S- \alpha \beta D, \delta S - C$) \\ 
        \hline
    \end{tabular}
\end{table}

\begin{theorem} \label{theorem5}
    By adding a interference term $D$, the aggregator tends to choose a smaller $Z$.
\end{theorem}

\begin{proof}

    

     The indifference condition for $\mathcal{V}$ remains unchanged, as the introduction of $D$ does not change the utility of the validator. Therefore, when reaching equilibrium, $\beta = \frac{C}{\delta S - T}$.

     The expected utility of the aggregator is consist of the expected utility of attacking and not attacking, that is $E_A=\beta((1-\alpha)Z+\alpha(-S-\alpha \beta D)) + (1-\beta)B$. 
     
     Fix $\alpha$ and other constants, the aggregator chooses $\beta$ that maximizes $E_A$, so $E_A'(\beta)=(-\alpha S - 2 {\alpha}^2 \beta D) + (1-\alpha)Z - B = 0$. That is,
     \begin{equation}
    \begin{aligned}
        {\alpha}^2 {\beta} D = \frac{(1-\alpha) Z - B - \alpha S}{2}.
        \label{abd1}
    \end{aligned}
    \end{equation}

    So the expected utility $E_A$ can be simplified to 
    \begin{equation}
    \begin{aligned}
        E_A(\beta)&=B-\frac{1}{2}\beta B-\frac{1}{2}\alpha\beta S+\frac{1}{2}(1-\alpha)\beta Z\\
        &= \frac{1}{2}\beta(2\alpha^2 \beta D-B)+B-\frac{1}{2}\beta B
        \label{eq_eaz}
    \end{aligned}
    \end{equation}

    Based on conditions above, we then study the optimal $Z$ that the aggregator chooses. Now $\alpha$ is a function of $Z$ (while $\beta$ is still irrelevant with $Z$).
    
    Take the derivative of $E_A$ with respect to $Z$ in Eq.(\ref{eq_eaz}), we have $E_A'(Z)=2\alpha\beta^2 D {\alpha}'$

    Take the derivative with respect to $Z$ in Eq.(\ref{abd1}), we have $4\alpha\beta D \alpha' + \alpha' S = 1-\alpha - \alpha'Z $, which implies $ \alpha' = \frac{1-\alpha}{4\alpha\beta D + S + Z} > 0$.

    Combining with the above conclusions, we can determine the positive or negative value of $E_A'(Z)$ by the positive or negative value of $D$. When we set $D$ to a negative value, $E_A'(Z) < 0$ holds. Therefore, the aggregator tends to choose a smaller Z, resulting in decreasing the system loss.
\end{proof}

It is worth noting that the system interference term \( D \) can be set to a very small value, so the introduction of \( D \) does not affect the behavior of other parties.

\section{Suggestion}
Based on the discussion above, we propose several suggestions for the system to maximize its benefits. These recommendations aim to optimize system parameters and mitigate potential losses from malicious behavior by aggregators.

\subsubsection{Subsidizing caught aggregators} For the system, we suggest that the system add an interference term $D$ providing subsidy to the caught aggregators. Because the benefits of the aggregator are independent of $Z$, while system loss $\mathcal{L}$ are positively correlated with $Z$. By introducing $D$, the aggregators prefer to choose a smaller $Z$ explained in Section 6.

\begin{figure}[h]
    \begin{minipage}{0.32\linewidth}
		\vspace{3pt}
		\centerline{\includegraphics[width=\textwidth]{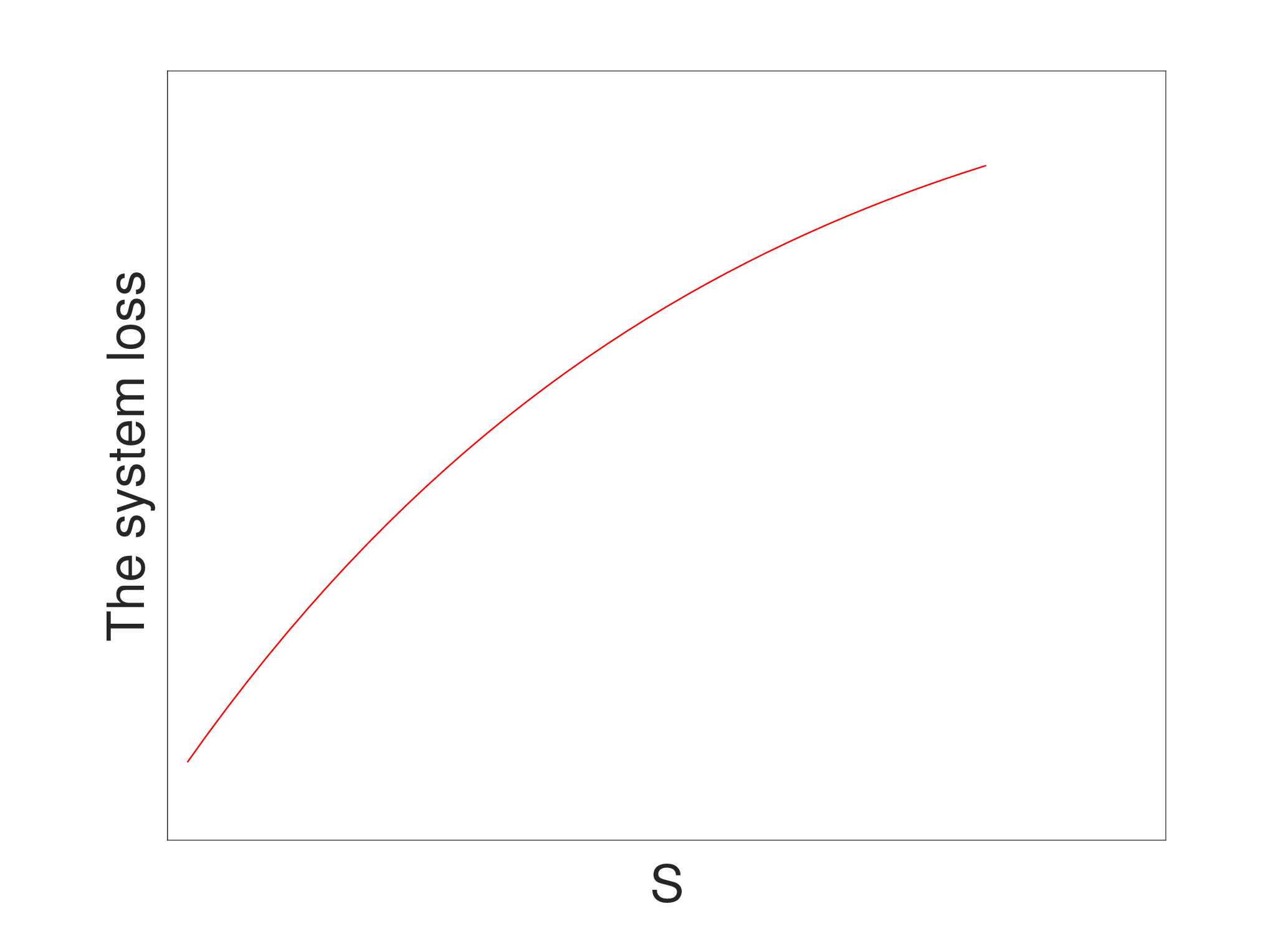}}
		\caption{The system loss $\mathcal{L}$ changes by $S$}
            \label{lossbys}
	\end{minipage}
	\begin{minipage}{0.32\linewidth}
		\vspace{3pt}
		\centerline{\includegraphics[width=\textwidth]{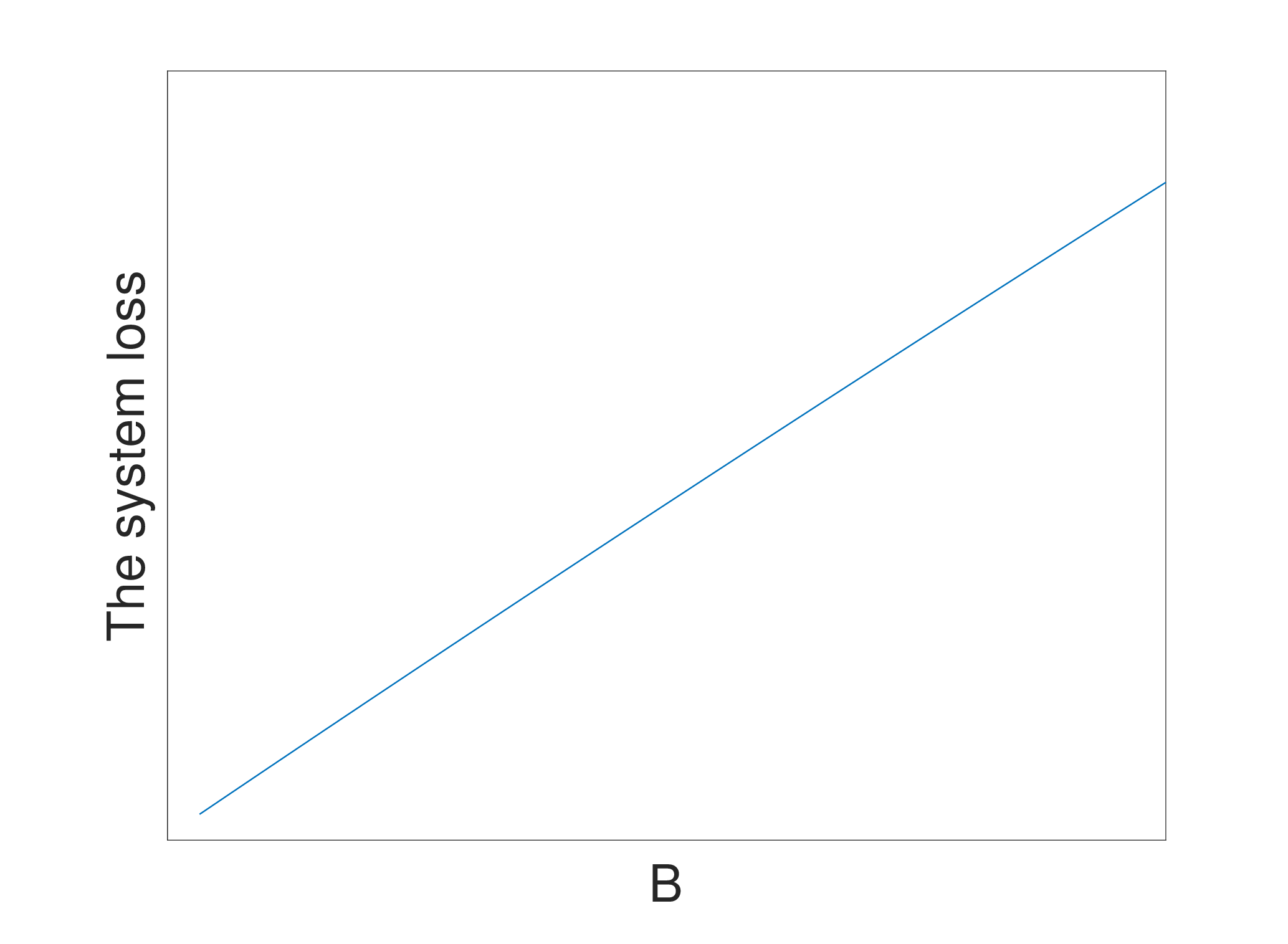}}
            \caption{The system loss $\mathcal{L}$ changes by $B$}
            \label{lossbyb}
	\end{minipage}
	\begin{minipage}{0.32\linewidth}
		\vspace{3pt}
		\centerline{\includegraphics[width=\textwidth]{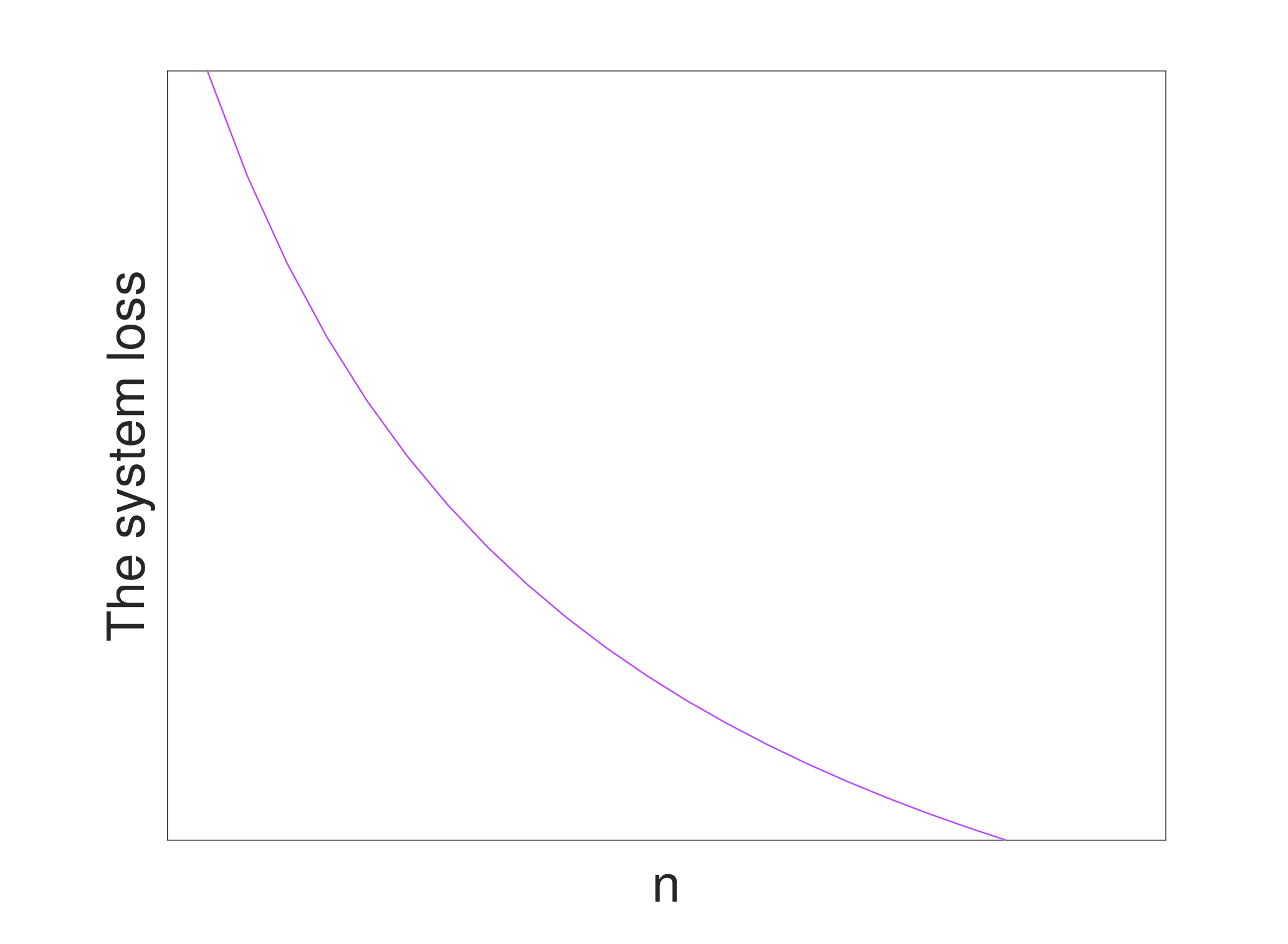}}
            \caption{The system loss $\mathcal{L}$ changes by $n$}
	    \label{lossbyn}
	\end{minipage}
\end{figure}

\subsubsection{Decrease $R$} In the previous section, we discussed the equilibria of games with multiple validators. We observe the rise of more better equilibria characterized by smaller system loss as $R$ decreases. Therefore, we consider reducing the value of $R$ from multiple aspects.

\subsubsection{Decerase $S$} Considering the scenario of multiple validators, the worst system loss $\mathcal{L}$ of multiple validators is, 
\begin{equation}
    \mathcal{L} = \frac{C}{P_n \delta S - Q_n (f_pV + \frac{T}{n}) + f_pV} \cdot \frac{(S+B)Z}{S+Z}.
    \label{lossnv}
\end{equation} 

Treating the system loss $\mathcal{L}$ as a function of deposit $S$, we make a graph of the system loss $\mathcal{L}$ as $S$ changes, shown in Figure.~\ref{lossbys}. We can learn that as $S$ increases, the system loss $\mathcal{L}$ increases. Therefore, we can decrease $S$ in order to decrease the system loss of the worst equilibrium.

It is worth noting that though decreasing $S$ will increase $R$, their  optimization directions are different: decreasing $R$ raises more better equilibrium, while decreasing $S$ reduces the system loss of the worst case equilibrium.

\subsubsection{Decrease $B$} Since \( P_n \) and \( Q_n \) are related to \( B \), we can treat the system loss \( \mathcal{L} \) as a function of the reward \( B \). Figure~\ref{lossbyb} illustrates how the system loss changes as \( B \) varies. By decreasing \( B \), we can reduce the system loss in the worst equilibrium.

\subsubsection{Increase $n$} Increasing the number of validators \( n \) will also affect the values of \( P_n \) and \( Q_n \), thereby increasing \( \Gamma_n \). This adjustment is more conducive to achieving more better equilibrium. Figure~\ref{lossbyn} shows the change in system loss with varying \( n \). Increasing \( n \) can also reduce the system loss in the worst equilibrium. Additionally, this also helps decrease $R$, resulting in better equilibria.

\subsubsection{Decrease $C$} From Eq.~(\ref{lossnv}), it is evident that decreasing the verification cost \( C \) for the verifier also helps to decrease the system loss, which demonstrates that technological progress can improve productivity.

\subsubsection{Increase $V$ and $f_p$} Since $Q_n<1$, $\mathcal{L}$ in Eq.\ref{lossnv} is decreasing with $f_pV$. Similar to decrease $S$, increasing \( f_pV \) also raises \( R \), resulted in different optimization directions.


\subsubsection{Increase $\delta$} Increasing \( \delta \) helps to decrease the worst system loss. At the same time, increasing \( \delta \) can decrease \( R \), raising more better equilibria.

\subsubsection{Decrease $T$} $\mathcal{L}$ in Eq.\ref{lossnv} is decreasing with $T$. Therefore, decreasing \( T \) can reduce the system loss \( \mathcal{L} \). Additionally, this also helps decrease \( R \), raising more better equilibria.

%
%
%
%

\section*{Appendix}
\subsubsection{The proof of Lemma~\ref{lemma1}}
\begin{proof}
    For any role in the game, when he chooses a behavior of his action space as pure-strategy player, he can always achieve higher return by changing his behavior.

    For the aggregator, when $\mathcal A$ is the pure-strategy of attacking, $\mathcal V$ is the pure-strategy of chance-taker when reaching maximum revenue. In this case, $\mathcal A$ will get a higher return when he changes his behavior to not attack. when $\mathcal A$ chooses not to attack, $\mathcal V$ is the pure-strategy of free-rider when reaching maximum revenue. $\mathcal A$ will get a higher return when he changes his behavior to attack.

    For the validators, when $\mathcal{V}$ is the honest verifier, $\mathcal{A}$ maximum its revenue when not attacking. Then $\mathcal{V}$ gets a higher revenue when changing his behavior to free-rider. But when $\mathcal{V}$ is the free-rider, $\mathcal{A}$ maximum its revenue when attacking. Then $\mathcal{V}$ gets a higher revenue when changing his behavior to chance-taker. When $\mathcal{V}$ is the chance-taker, $\mathcal{A}$ maximum its revenue when not attacking. Then $\mathcal{V}$ gets a higher revenue when changing his behavior to free-ride.
    
    Above all, there is no Nash Equilibrium when both $\mathcal{A}$ and $\mathcal{V}$ play the pure strategy.
\end{proof}

\subsubsection{The proof of Theorem~\ref{theorem1}}
\begin{proof}
    Since there is no pure-strategy equilibrium in the game, we consider mixed strategy equilibrium.

    First, we consider the case that $\mathcal A$ plays the mixed strategy that he attacks with a probability $\beta$ and $\mathcal V$ plays the mixed strategy that he chooses to be a chance-taker or a free-rider. Table~\ref{table:matrixalpha0} shows the payoff matrix in this case.

    \begin{table}[!ht]
        \centering
        \caption{Payoff Matrix When $\alpha$ = 0}
        \label{table:matrixalpha0}
        \begin{tabular}{|l|l|l|}
            \hline
             & $\overline A$ $(1 - \beta)$ & A $(\beta)$ \\ 
             \hline
            $\overline {VC}$ $(1-\gamma)$ & (B, T) & (Z, T) \\ 
            $\overline V C$  ($\gamma$) & $(B+\lambda f_n V, -f_n V)$ & $(-S, \delta S)$ \\ 
            \hline
        \end{tabular}
    \end{table}

    The indifference condition for $\mathcal A$ is that the expected utility of attacking is equal to the expected utility of not attacking. That is,

    \begin{equation*}
    \begin{aligned}
        \gamma(-S)+(1-\gamma)Z = \gamma(B+\lambda f_n V)+(1-\gamma)B, 
    \end{aligned}
    \end{equation*}

    equivalent to
    \begin{equation*}
    \begin{aligned}
        \gamma = \frac{Z-B}{Z+S+\lambda f_n V}.
    \end{aligned}
    \end{equation*}

    The indifference condition for $\mathcal V$ is that the expected utility of chance-taker is equal to the expected utility of free-rider. That is,

    \begin{equation*}
    \begin{aligned}
        (1-\beta)(-f_n V)+\beta \delta S = T,
    \end{aligned}
    \end{equation*}

    equivalent to
    \begin{equation*}
    \begin{aligned}
        \beta = \frac{T+f_n V}{\delta S+f_n V}.
    \end{aligned}
    \end{equation*}

    And both utilities are greater than the utility of honest verifier, that is, 

    \begin{equation*}
    \begin{aligned}
        T \ge (1-\beta)(T-C)+\beta(\delta S-C),
    \end{aligned}
    \end{equation*}

    equivalent to
    \begin{equation*}
    \begin{aligned}
        C \ge \frac{(\delta S-T)(T+f_n V)}{\delta S+f_n V}.
    \end{aligned}
    \end{equation*}

    So when $\displaystyle C \ge \frac{(\delta S-T)(T+f_n V)}{\delta S+f_n V}$, there is a Nash Equilibrium that $\mathcal A$ attacks or not with a probability $\displaystyle \beta = \frac{T+f_n V}{\delta S+f_n V}$ and $\mathcal V$ plays as a chance-taker with a probability $\displaystyle \gamma = \frac{Z-B}{Z+S+\lambda f_n V}$ or a free-rider with a probability $1-\gamma$.

    Also, we reach a similar conclusion when $\mathcal V$ plays the mixed strategy that he chooses to be a free-rider or a verifier. Table~\ref{table:matrixgamma0} shows the payoff matrix in this case.

     \begin{table}[!ht]
        \centering
        \caption{Payoff Martix When $\gamma$ = 0}
        \label{table:matrixgamma0}
        \begin{tabular}{|l|l|l|}
            \hline
            & $\overline A$ $(1 - \beta)$ & A $(\beta)$ \\ 
            \hline
            $\overline {VC}$ $(1-\alpha)$ & $(B, T)$ & $(Z, T)$ \\ 
            V $(\alpha)$ & $(B, T-C)$ & $(-S, \delta S - C$) \\ 
            \hline
        \end{tabular}
    \end{table}

    The indifference condition for $\mathcal A$ is that the expected utility of attacking is equal to the expected utility of not attacking. That is,

    \begin{equation*}
    \begin{aligned}
        \alpha(-S)+(1-\alpha)Z=B, 
    \end{aligned}
    \end{equation*}

    equivalent to
    \begin{equation*}
    \begin{aligned}
        \alpha = \frac{Z-B}{Z+S}.
    \end{aligned}
    \end{equation*}

    The indifference condition for $\mathcal V$ is that the expected utility of verifier is equal to the expected utility of free-rider. That is,

    \begin{equation*}
    \begin{aligned}
        \beta(\delta S - C)+(1-\beta)(T -C) = T,
    \end{aligned}
    \end{equation*}

    equivalent to
    \begin{equation*}
    \begin{aligned}
        \beta = \frac{C}{\delta S-T}.
    \end{aligned}
    \end{equation*}

    And both utilities are greater than the utility of honest verifier, that is, 

    \begin{equation*}
    \begin{aligned}
        T \ge (1-\beta)(-f_n V)+\beta \delta S,
    \end{aligned}
    \end{equation*}

    equivalent to
    \begin{equation*}
    \begin{aligned}
        C \le \frac{(\delta S-T)(T+f_n V)}{\delta S+f_n V}.
    \end{aligned}
    \end{equation*}

    So when $\displaystyle C \le \frac{(\delta S-T)(T+f_n V)}{\delta S+f_n V}$, there is a Nash Equilibrium that $\mathcal A$ attacks or not with a probability $\displaystyle \beta = \frac{C}{\delta S-T}$ and $\mathcal V$ plays as a verifier with a probability $\displaystyle \alpha = \frac{Z-B}{Z+S}$ or a free-rider with a probability $1-\alpha$.

    Especially, when $\displaystyle C=\frac{(\delta S-T)(T+f_n V)}{\delta S+f_n V}$, there is an additional Nash Equilibrium. Table~\ref{table:matrixall} shows the payoff matrix with all supports.

    \begin{table}[!ht]
        \centering
        \caption{Payoff Matrix With all supports}
        \begin{tabular}{|l|l|l|}
            \hline
            & $\overline A$ $(1 - \beta)$ & A $(\beta)$ \\ 
            \hline
            $\overline {VC}$ $(1-\alpha-\gamma)$ & $(B, T)$ & $(Z, T)$ \\ 
            $\overline VC$  ($\gamma$) & $(B+\lambda f_n V, -f_n V)$ & $(-S, \delta S)$ \\ 
            V $(\alpha)$ & $(B, T-C)$ & $(-S, \delta S - C)$ \\ 
            \hline
        \end{tabular}
        \label{table:matrixall}
    \end{table}
\end{proof} 

\subsubsection{The proof of Theorem~\ref{theorem2}}
\begin{proof}
    Suppose there are two validators, both of which choose to play a mixed strategy that $\mathcal V_1$/$\mathcal V_2$ verifies with a probability $\alpha_1$/$\alpha_2$, or with a probability $(1-\alpha_1)$/$(1-\alpha_2)$ as a free-rider. Table~\ref{table:kyc1} presents a payoff matrix that summarizes all players’ payoffs from their strategic interactions as described above.

    \begin{table}[ht]
    \centering
    \caption{Payoff matrix $(\mathcal A, \mathcal V_1, \mathcal V_2)$ when $\mathcal V_1$ and $\mathcal V_2$ play a mixed strategy of free-rider and verifier with a probability $\alpha_1/\alpha_2$}
    \renewcommand\arraystretch{1.5}
    \begin{tabular}{|l|l|l|}
        \hline
         & $\overline A$ $(1 - \beta)$ & A $(\beta)$ \\ 
         \hline
         $(V_1(\alpha_1), V_2(\alpha_2))$ & $(B, \frac{T}{2}-C, \frac{T}{2}-C)$ & $(-S, \frac{\delta S}{2}-C, \frac{\delta S}{2}-C)$ \\ 
        $(\overline{V_1} (1-\alpha_1), V_2(\alpha_2))$ & $(B, \frac{T}{2}, \frac{T}{2}-C)$ & $(-S, -f_p V, \delta S -C)$ \\ 
        $(V_1(\alpha_1), \overline{V_2} (1-\alpha_2))$ & $(B, \frac{T}{2}-C, \frac{T}{2})$ & $(-S, \delta S -C, -f_p V)$ \\ 
        $(\overline {V_1}(1-\alpha_1), \overline {V_2}(1-\alpha_2)) $ & $(B, \frac{T}{2}, \frac{T}{2})$ & $(Z, \frac{T}{2}, \frac{T}{2})$ \\ 
        \hline
    \end{tabular}
    \label{table:kyc1}
    \end{table}

    The indifference condition for $\mathcal A$ is that the expected utility of attacking is equal to the expected utility of not attacking. That is,

    \begin{equation}
    \begin{aligned}
        (1-\alpha_1)(1-\alpha_2)Z - (1-(1-\alpha_1)(1-\alpha_2))S = B.
        \label{eq2ea}
    \end{aligned}
    \end{equation}

    The indifference condition for $\mathcal V_1$ is that the expected utility of verifier is equal to the expected utility of free-rider. That is,

    \begin{equation}
    \begin{aligned}
        \beta(\alpha_2 \frac{1}{2}\delta S + (1-\alpha_2) \delta S) + (1-\beta)(\frac{1}{2}T) -C = (1-\alpha_2 \beta) \frac{1}{2}T-\alpha_2 \beta f_p V.
        \label{eq2ev1}
    \end{aligned}
    \end{equation}

    Similarly, the indifference condition for $\mathcal V_2$ is that the expected utility of verifier is equal to the expected utility of free-rider. That is,

    \begin{equation}
    \begin{aligned}
        \beta(\alpha_1 \frac{1}{2}\delta S + (1-\alpha_1) \delta S) + (1-\beta)(\frac{1}{2}T) -C = (1-\alpha_1 \beta) \frac{1}{2}T-\alpha_1 \beta f_p V.
        \label{eq2ev2}
    \end{aligned}
    \end{equation}

    By combining Eq.(\ref{eq2ea}), Eq.(\ref{eq2ev1}) and Eq.(\ref{eq2ev2}), we can obtain, 

    \begin{equation*}
    \begin{aligned}
        \alpha_1 = \alpha_2 = 1-\sqrt{\frac{B+S}{Z+S}}\\
        \beta_1 = \frac{C}{\delta S (1-\frac{1}{2} \alpha) + \alpha(\frac{1}{2}T+f_pV) - \frac{1}{2}T}
    \end{aligned} 
    \end{equation*}

    So in any conditions, there is a Nash Equilibrium that ($\mathcal V_1$/$\mathcal V_2$) plays the mixed strategy that he verifies with a probability ($\alpha_1$/$\alpha_2$) and $\mathcal A$ attacks with probability $\beta$.

     Then we consider the case that one of the validators plays as a free-rider. Suppose $\mathcal V_2$ is the free-rider, Table~\ref{table:kyc2} presents a payoff matrix that summarizes all players’ payoffs from their strategic interactions.

    \begin{table}[ht]
    \centering
    \caption{Payoff matrix $(\mathcal A, \mathcal V_1, \mathcal V_2)$ when $\mathcal V_1$ plays a mixed strategy of free-rider and verifier with a probability $\alpha$ and $\mathcal V_2$ plays as a free-rider}
    \renewcommand\arraystretch{1.5}
    \begin{tabular}{|l|l|l|}
        \hline
         & $\overline A$ $(1 - \beta)$ & A $(\beta)$ \\ 
         \hline
        $(V_1(\alpha), \overline{V_2})$ & $(B, \frac{T}{2}-C, \frac{T}{2})$ & $(-S, \delta S -C, -f_p V)$ \\ 
        $(\overline {V_1}(1-\alpha), \overline {V_2}) $ & $(B, \frac{T}{2}, \frac{T}{2})$ & $(Z, \frac{T}{2}, \frac{T}{2})$ \\ 
        \hline
    \end{tabular}
    \label{table:kyc2}
    \end{table}
    
    Similarly, the indifference condition for $\mathcal A$ is that the expected utility of attacking is equal to the expected utility of not attacking. That is,

    \begin{equation}
    \begin{aligned}
        (1-\alpha)Z - \alpha S = B.
        \label{eq3ea}
    \end{aligned}
    \end{equation}

    The indifference condition for $\mathcal V_1$ is that the expected utility of verifier is equal to the expected utility of free-rider. That is,

    \begin{equation}
    \begin{aligned}
        \beta\delta S + (1-\beta)(\frac{1}{2}T) -C = \frac{1}{2}T.
        \label{eq3ev1}
    \end{aligned}
    \end{equation}

    The indifference condition for $\mathcal V_2$ is that the expected utility of free-rider is greater than the expected utility of verifier. That is,

    \begin{equation}
    \begin{aligned}
        (1-\alpha\beta)\frac{T}{2} - \beta\alpha f_p V \ge
        \beta(\alpha \frac{1}{2}\delta S + (1-\alpha) \delta S) + (1-\beta)(\frac{1}{2}T) -C 
        \label{eq3ev2}
    \end{aligned}
    \end{equation}

    By combining Eq.(\ref{eq3ea}), Eq.(\ref{eq3ev1}) and Eq.(\ref{eq3ev2}), we can obtain, 

    \begin{equation*}
      \begin{aligned}
        \alpha &= \frac{Z-B}{Z+S}\\
        \beta_2 &= \frac{C}{\delta S - \frac{1}{2}T} \\
        \frac{1}{2} &\ge \frac{ \frac{1}{2}T+f_pV}{\delta S} 
      \end{aligned}
    \end{equation*}

     So when $\displaystyle R =  \frac{ \frac{1}{2}T+f_pV}{\delta S} \le \frac{1}{2}$, there is an equilibrium that $\mathcal V_1$ plays the mixed strategy that he verifies with a probability $\alpha_1$ or not, and $\mathcal V_2$ plays as a free-rider.  
\end{proof}

\subsubsection{The proof of Lemma~\ref{lemma2}}
\begin{proof}
    Take any two validators among the $m$ mixed strategy validators, denoted as $V_1$ and $V_2$ respectively, then the probabilities of verifying are $\alpha_1$ and $\alpha_2$ respectively. Then denote $P_i$ as the probability that there are $i$ validators verifying among the remaining $m-2$ validators. The validator $V_1$ is indifferent between honest verifier and free-rider, that is,

    \begin{equation*}
    \begin{aligned}
        &\beta[\alpha_2 \sum_{k=0}^{m-2}P_k\frac{\delta S}{k+2}+(1-\alpha_2)\sum_{k=0}^{m-2}P_k\frac{\delta S}{k+1}]+(1-\beta)\frac{T}{n} -C \\
        = &  \beta(1-P_0(1-\alpha_2))(-f_pV) + (1-\beta(1-P_0(1-\alpha_2)))\frac{T}{n},
    \end{aligned}
    \label{eqev1}
    \end{equation*}

    which simplifies to 

    \begin{equation}
    \begin{aligned}
       & \beta[\sum_{k=0}^{m-2}P_k\frac{\delta S}{k+2}-\sum_{k=0}^{m-2}P_k\frac{\delta S}{k+1}+P_0f_pV+P_0\frac{T}{n}]\alpha_2 \\
       =& C-\beta [\sum_{k=0}^{m-2}\frac{P_k}{k+1}\delta S - P_0(f_pV+\frac{T}{n}) + f_pV]\\
    \end{aligned}
    \label{eqev1s}
    \end{equation}
    
    Eq.(\ref{eqev1s}) is also to hold for the validator $V_2$ and other validators.

    Denote $g(\alpha_3,\alpha_4, \dots, \alpha_m)$ as the right side of Eq.(\ref{eqev1s}). Then for $\forall (m-2)$ validators among $m$ validators, $g(\alpha_{i_1},\alpha_{i_2}, \dots, \alpha_{i_{m-2}})=0$ is satisfied when the two validators' probability $\alpha_1 \ne \alpha_2$. We assume that there are only two probabilities among m validators. A certain number of $k$ validators in $m$ validators verify with probability $\alpha_1$, and the remaining $m-k$ validators in $m$ validators verify with probability $\alpha_2$. That is, for anyone of $m$ validators, the remaining $m-1$ validators must contain a subset of $m-2$ validators such that there are $k-1$ mixed strategy validators with probability $\alpha_1$ and $m-k-1$ mixed strategy validators with probability $\alpha_2$. Especially, when $\alpha_1 = \alpha_2$, it is the case that all mixed strategy validators are symmetric.

    Now we are considering situations where there are more than two probabilities among m validators. Assume that there are an another probability $\alpha_3$. We define that a certain number of $k$ validators in $m$ validators verify with probability $\alpha_1$, a certain number of $l$ validators in $m$ validators verify with probability $\alpha_2$, and the remaining $m-k-l$ validators in $m$ validators verify with probability $\alpha_3$. And 
$\alpha_1 \ne \alpha_2 \ne \alpha_3$. 

    The validator with probability $\alpha_1$ is indifferent between honest verifier and free-rider, that is,

    \begin{equation}
    \begin{aligned}
        &\beta[\alpha_2\alpha_3 \sum_{k=0}^{m-3}P_k\frac{\delta S}{k+3}+(1-\alpha_2)\alpha_3\sum_{k=0}^{m-3}P_k\frac{\delta S}{k+2}\\
        &+(1-\alpha_3)\alpha_2\sum_{k=0}^{m-3}P_k\frac{\delta S}{k+2} + (1-\alpha_2)(1-\alpha_3)\sum_{k=0}^{m-3}P_k\frac{\delta S}{k+1}]+(1-\beta)\frac{T}{n} -C \\
        = &\beta(1-P_0(1-\alpha_2)(1-\alpha_3)(-f_pV) + (1-\beta(1-P_0(1-\alpha_2)(1-\alpha_3))\frac{T}{n}
    \end{aligned}
    \label{eqe123}
    \end{equation}

    Denote $\displaystyle \sum_{k=0}^{m-3}\frac{P_k}{k+1} = E_1$, $\displaystyle \sum_{k=0}^{m-3}\frac{P_k}{k+2} = E_2$, $\displaystyle \sum_{k=0}^{m-3}\frac{P_k}{k+3} = E_3$, Eq.(\ref{eqe123}) can be simplified as:

    \begin{equation}
    \begin{aligned}
        &[(E_2-E_1)\delta S + P_0(f_pV+\frac{T}{n})]\alpha_2 + [(E_2-E_1)\delta S + P_0(f_pV+\frac{T}{n})]\alpha_3 \\
        +& [(E_3-2E_2+E_1)\delta S - P_0(f_pV+\frac{T}{n})]\alpha_2\alpha_3\\
        =&\frac{C}{\beta} - [f_pV+P_0(f_pV+\frac{T}{n})+E_1\delta S]
    \end{aligned}
    \label{eqe123s}
    \end{equation}

    Similarly, for someone with probability $\alpha_2$, there is:
    \begin{equation}
    \begin{aligned}
        &[(E_2-E_1)\delta S + P_0(f_pV+\frac{T}{n})]\alpha_1 + [(E_2-E_1)\delta S + P_0(f_pV+\frac{T}{n})]\alpha_3 \\
        +& [(E_3-2E_2+E_1)\delta S - P_0(f_pV+\frac{T}{n})]\alpha_1\alpha_3\\
        =&\frac{C}{\beta} - [f_pV+P_0(f_pV+\frac{T}{n})+E_1\delta S]
    \end{aligned}
    \label{eqe13s}
    \end{equation}

    For $\alpha_1 \ne \alpha_2$, Eq.(\ref{eqe123s}) and Eq.(\ref{eqe13s}) hold. So

    \begin{equation}
    \begin{aligned}
        [(E_2-E_1)\delta S + P_0(f_pV+\frac{T}{n})] + [(E_3-2E_2+E_1)\delta S - P_0(f_pV+\frac{T}{n})]\alpha_3 = 0
    \end{aligned}
    \label{eqalpha}
    \end{equation}

    Eq.(\ref{eqalpha}) holds for $\alpha_1$, $\alpha_2$ and $\alpha_3$, so:
    \begin{equation}
    \left\{
    \begin{aligned}
        [(E_2-E_1)\delta S + P_0(f_pV+\frac{T}{n})] + [(E_3-2E_2+E_1)\delta S - P_0(f_pV+\frac{T}{n})]\alpha_1 = 0 \\
        [(E_2-E_1)\delta S + P_0(f_pV+\frac{T}{n})] + [(E_3-2E_2+E_1)\delta S - P_0(f_pV+\frac{T}{n})]\alpha_2 = 0\\
        [(E_2-E_1)\delta S + P_0(f_pV+\frac{T}{n})] + [(E_3-2E_2+E_1)\delta S - P_0(f_pV+\frac{T}{n})]\alpha_3 = 0
    \end{aligned}
    \right.
    \end{equation}

    So we can learn that:

    \begin{equation}
    \left\{
    \begin{aligned}
        (E_2-E_1)\delta S + P_0(f_pV+\frac{T}{n}) = 0 \\
        (E_3-2E_2+E_1)\delta S - P_0(f_pV+\frac{T}{n}) = 0
    \end{aligned}
    \right.
    \label{eqe2e3}
    \end{equation}

    When Eq.(\ref{eqe2e3}) holds, we can refer that $E_2 = E_3$. But by definition, $E_2$ and $E_3$ are not equal, so there are three different probabilities among m validators that do not hold true.
\end{proof}

\subsubsection{The proof of Lemma~\ref{lemma3}}
\begin{proof}
The validator of mixed strategy is indifferent between honest verifier and free-rider, that is, 
\begin{equation*}
\begin{aligned}
    \beta [\sum _{k=0} ^{m-1} {\frac{P_k}{k+1}} ]\delta S + (1-\beta)\frac{T}{n} -C = (1-\beta(1-P_0)) \frac{T}{n} - \beta (1-P_0) f_p V.
\end{aligned}
\end{equation*}

which simplifies to 
\begin{equation}
\begin{aligned}
    \beta(\sum _{k=0} ^{m-1} \frac{P_k}{k+1} \delta S - \frac{P_0T}{n}+(1-P_0)f_p V) = C. \label{eqem}
\end{aligned}
\end{equation}

Similarly, we consider one of the n-m validators that play the pure strategy as the free-rider. The expected utility of free-rider behavior is greater than that of verifier behavior, that is 
\begin{equation*}
\begin{aligned}
    &(1-\beta(1-P_0(1-\alpha))) \frac{T}{n} - \beta (1-P_0(1-\alpha)) f_p V \\
    \ge& \beta  [\sum _{k=0} ^{m-1} {\frac{P_k \alpha}{k+2}} + \sum _{k=0} ^{m-1} {\frac{P_k (1-\alpha)}{k+1}}]\delta S + (1-\beta)\frac{T}{n} -C,
\end{aligned}
\label{eqen-m}
\end{equation*}

which simplifies to 
\begin{equation}
\begin{aligned}
    \beta(\sum _{k=0} ^{m-1} [\frac{P_k\alpha}{k+2} + \frac{P_k(1-\alpha)}{k+1}]\delta S - \frac{P_0T(1-\alpha)}{n} + (1-P_0(1-\alpha))f_p V) \le C \label{eqea}
\end{aligned}
\end{equation}

Assuming that all validators play the mixed strategy, there is no free-rider among the validators, so the condition that the expected utility of free-rider is greater than that of verifier does not hold for the case $m=n$. 

Considering the symmetric case where all mixed strategy validators adopt the same strategy with probability $\alpha_m$, we can simplify Eq.~(\ref{eqem}) to:

\begin{equation}
\begin{aligned}
    \beta (P_m \delta S - Q_m (f_pV+\frac{T}{n}) + f_pV) = C.
    \label{em}
\end{aligned}
\end{equation}

In a similar manner, Eq.~(\ref{eqea}) can be simplified to:
\begin{equation}
\begin{aligned}
    \beta(\frac{1-A(1-\alpha_m)}{(m+1)\alpha_m}\delta S - A(f_pV + \frac{T}{n}) + f_pV) \le C \label{en-m}.
\end{aligned}
\end{equation}

By combining Eq.~(\ref{em}) and Eq.~(\ref{en-m}), we derive the following inequality:
\begin{equation}
\left[\frac{1}{m(m+1)}\left(\frac{1}{A}-1\right)-\frac{\alpha_m}{m+1}\right]\frac{1-\alpha_m}{\alpha_m^2} \delta S \ge f_pV+\frac{T}{n}.
\label{gammam}
\end{equation}

Since \(\displaystyle \Gamma_m = \left[\frac{1}{m(m+1)}\left(\frac{1}{A}-1\right)-\frac{\alpha_m}{m+1}\right]\frac{1-\alpha_m}{\alpha_m^2}\), then Eq.~(\ref{gammam}) can be simplified to:

\begin{equation*}
R \le \Gamma_m.
\end{equation*}

Also, the aggregator is indifferent between attacking and not attacking, that is, 

\begin{equation*}
\begin{aligned}
    P_0(1-\alpha_m) Z - (1-P_0(1-\alpha_m))S  = B,
\end{aligned}
\end{equation*}

which simplifies to 
\begin{equation*}
\begin{aligned}
    (1-\alpha_m)^m = \frac{B+S}{Z+S} = A.
\end{aligned}
\end{equation*}

\end{proof}

\subsubsection{The proof of Theorem~\ref{main_theorem}}
\begin{proof}
    Based on Lemma~\ref{lemma_Gamma}, if $R < \Gamma_m$ holds for a given $m$, it also holds that $R < \Gamma_{m+1}$. This leads to the inequality $R < \Gamma_m < \Gamma_{m+1} < \dots < \Gamma_n$, indicating that if an $m$-NE exists, then $(m+1)$-NE, ..., $n$-NE also exist.

    Then we search for the optimal equilibrium. To minimize the system loss \(L\), we need to minimize the value of \(\beta\). Suppose that when there are \(m\) mixed strategy validators, \(\beta\) attains its minimum value, denoted as \(\beta_m\). Therefore, we have:
    \begin{equation*}
    \begin{aligned}
    \beta_m \le \beta_{m+1} & \iff P_m \delta S - Q_m (f_pV+\frac{T}{n}) \ge P_{m+1} \delta S - Q_{m+1} (f_pV+\frac{T}{n}) \\
    & \iff \frac{\frac{T}{n} + f_pV}{\delta S} \le \frac{P_m - P_{m+1}}{Q_m - Q_{m+1}} = \Delta_m.
    \end{aligned}
    \end{equation*}
    
    Similarly, we have:
    \begin{equation*}
    \begin{aligned}
    \beta_m \le \beta_{m-1} & \iff P_m \delta S - Q_m (f_pV+\frac{T}{n}) \le P_{m-1} \delta S - Q_{m-1} (f_pV+\frac{T}{n}) \\ 
    & \iff \frac{\frac{T}{n} + f_pV}{\delta S} \ge \frac{P_{m-1} - P_m}{Q_{m-1} - Q_m} = \Delta_{m-1}.
    \end{aligned}
    \end{equation*}
    
    When $\Delta_{m-1} < R < \Delta_m$, $\beta_m$ reaches its minimum value, resulting in the minimum system loss.

    \begin{figure*}[htp]
        \centering
        \includegraphics[scale=0.3]{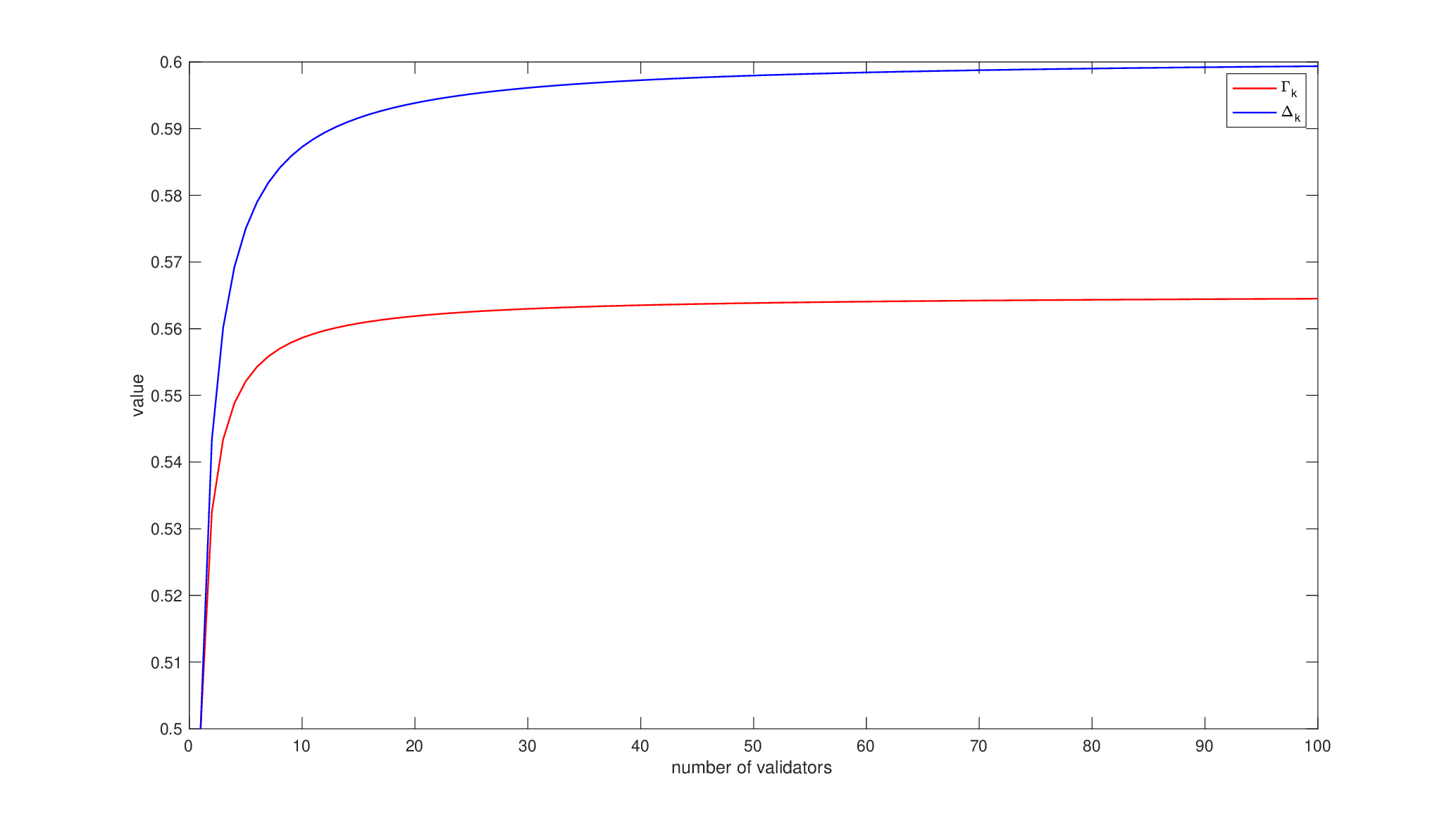}
        \caption{The image of \textcolor{blue}{$\Delta_m$} and \textcolor{red}{$\Gamma_m$}. The x-coordinate represents the number of mixed strategy validators. We set A to 0.7}
        \label{fig:delta_gamma}
    \end{figure*}
    
    If \(R\) lies between \(\Gamma_{m-1}\) and \(\Gamma_m\), all possible equilibriums consist of the following $n-m$ equilibrium: $m$-NE, $(m+1)$-NE, ..., $n$-NE. Among these $n-m$ equilibriums, the $m$-NE is optimal, as \(R < \Delta_m < \cdots < \Delta_n\). This ordering implies that \(\beta_k = \beta_m < \beta_{m+1} < \cdots < \beta_n\).

    This completes the proof of Theorem~\ref{main_theorem}.
\end{proof}

\subsubsection{The proof of Lemma~\ref{lemma7}}
\begin{proof}
For any of the \(k\) validators, they are indifferent between acting as the verifier and as the free-rider. This condition can be expressed as:
\begin{equation}
\begin{aligned}
    &\beta\left[\sum_{i=0}^{k-1} \binom{k-1}{i} \alpha_1^i (1-\alpha_1)^{k-1-i} 
    \sum_{j=0}^{m-k} \binom{m-k}{j} \alpha_2^j (1-\alpha_2)^{m-k-j} \frac{\delta S}{i+j+1}\right] + (1-\beta)\frac{T}{n} - C \\
    &= \beta\left[1-(1-\alpha_1)^{k-1}(1-\alpha_2)^{m-k}\right](-f_p V) + \left[1-\beta\left(1-(1-\alpha_1)^{k-1}(1-\alpha_2)^{m-k}\right)\right]\frac{T}{n}.
    \label{f21}
\end{aligned}
\end{equation}

Similarly, for the \(m-k\) validators, the condition of indifference between acting as a verifier and as a free-rider is given by:

\begin{equation}
\begin{aligned}
    &\beta\left[\sum_{i=0}^{k} \binom{k}{i} \alpha_1^i (1-\alpha_1)^{k-i} 
    \sum_{j=0}^{m-k-1} \binom{m-k-1}{j} \alpha_2^j (1-\alpha_2)^{m-k-1-j} \frac{\delta S}{i+j+1}\right] + (1-\beta)\frac{T}{n} - C \\
    &= \beta\left[1-(1-\alpha_1)^k(1-\alpha_2)^{m-k-1}\right](-f_p V) + \left[1-\beta\left(1-(1-\alpha_1)^k(1-\alpha_2)^{m-k-1}\right)\right]\frac{T}{n}.
    \label{f22}
\end{aligned}
\end{equation}

To simplify these equations, we introduce the following notation:

\begin{equation*}
\begin{aligned}
    p_1 &= \sum_{i=0}^{k-1} \binom{k-1}{i} \alpha_1^i (1-\alpha_1)^{k-1-i} 
    \sum_{j=0}^{m-k} \binom{m-k}{j} \alpha_2^j (1-\alpha_2)^{m-k-j} \frac{1}{i+j+1} \\
    &= \alpha_2\left(\sum_{i=0}^{k-1} \binom{k-1}{i} \alpha_1^i (1-\alpha_1)^{k-1-i} 
    \sum_{j=0}^{m-k-1} \binom{m-k-1}{j} \alpha_2^j (1-\alpha_2)^{m-k-1-j} \frac{1}{i+j+2}\right) \\
    &\quad + (1-\alpha_2)\left(\sum_{i=0}^{k-1} \binom{k-1}{i} \alpha_1^i (1-\alpha_1)^{k-1-i} 
    \sum_{j=0}^{m-k-1} \binom{m-k-1}{j} \alpha_2^j (1-\alpha_2)^{m-k-1-j} \frac{1}{i+j+1}\right),
\end{aligned}
\end{equation*}

\begin{equation*}
\begin{aligned}
    p_2 &= \sum_{i=0}^{k} \binom{k}{i} \alpha_1^i (1-\alpha_1)^{k-i} 
    \sum_{j=0}^{m-k-1} \binom{m-k-1}{j} \alpha_2^j (1-\alpha_2)^{m-k-1-j}\frac{1}{i+j+1} \\
    &= \alpha_1\left(\sum_{i=0}^{k-1} \binom{k-1}{i} \alpha_1^i (1-\alpha_1)^{k-1-i} 
    \sum_{j=0}^{m-k-1} \binom{m-k-1}{j} \alpha_2^j (1-\alpha_2)^{m-k-1-j} \frac{1}{i+j+2}\right) \\
    &\quad + (1-\alpha_1)\left(\sum_{i=0}^{k-1} \binom{k-1}{i} \alpha_1^i (1-\alpha_1)^{k-1-i} 
    \sum_{j=0}^{m-k-1} \binom{m-k-1}{j} \alpha_2^j (1-\alpha_2)^{m-k-1-j} \frac{1}{i+j+1}\right).
\end{aligned}
\end{equation*}

Let $\displaystyle p_3 = \sum_{i=0}^{k-1}C_{k-1}^i \alpha_1 ^i (1-\alpha_1)^{k-1-i} 
\sum_{j=0}^{m-k-1}C_{m-k-1}^j \alpha_2 ^j (1-\alpha_2)^{m-k-1-j} \frac{1}{i+j+2}$, 
$\displaystyle p_4 = \sum_{i=0}^{k-1}C_{k-1}^i \alpha_1 ^i (1-\alpha_1)^{k-1-i} 
\sum_{j=0}^{m-k-1}C_{m-k-1}^j \alpha_2 ^j (1-\alpha_2)^{m-k-1-j} \frac{1}{i+j+1}$.

Then, \(p_1\) and \(p_2\) can be simplified to:

\begin{equation*}
\begin{aligned}
    p_1 &= \alpha_2 p_3 + (1-\alpha_2) p_4, \\
    p_2 &= \alpha_1 p_3 + (1-\alpha_1) p_4.
\end{aligned}
\end{equation*}

Given that \((1-\alpha_1)^k (1-\alpha_2)^{m-k} = A\), let \(p_5 = (1-\alpha_1)^{k-1} (1-\alpha_2)^{m-k-1}\). Then, Eq.~(\ref{f21}) and Eq.~(\ref{f22}) can be simplified to:

\begin{subequations}
\begin{align}
    \beta\left[(\alpha_2 p_3 + (1-\alpha_2) p_4) \delta S + f_p V - (1-\alpha_2) p_5 (f_p V + \frac{T}{n})\right] = C, \label{eq38a} \\
    \beta\left[(\alpha_1 p_3 + (1-\alpha_1) p_4) \delta S + f_p V - (1-\alpha_1) p_5 (f_p V + \frac{T}{n})\right] = C. \label{eq38b}
\end{align}
\end{subequations}

Subtracting Eq.~(\ref{eq38b}) from Eq.~(\ref{eq38a}) yields:

\begin{equation}
\begin{aligned}
    \left[(p_3 - p_4) \delta S + p_5 (f_p V + \frac{T}{n})\right](\alpha_2 - \alpha_1) = \frac{C}{\beta} - p_4 \delta S - f_p V + p_5 (f_p V + \frac{T}{n}).
    \label{eqp}
\end{aligned}
\end{equation}

In the asymmetric case, we assume \(\alpha_1 \ne \alpha_2\), so:

\begin{equation*}
\begin{aligned}
    (p_3 - p_4) \delta S + p_5 (f_p V + \frac{T}{n}) = 0,  \\
    \frac{C}{\beta} - p_4 \delta S - f_p V + p_5 (f_p V + \frac{T}{n}) = 0. 
\end{aligned}
\end{equation*}

For the aggregator, indifference between attacking and not attacking is given by:

\begin{equation*}
\begin{aligned}
    (1-\alpha_1)^k (1-\alpha_2)^{m-k} Z - \left[1 - (1-\alpha_1)^k (1-\alpha_2)^{m-k}\right] S = B.
\end{aligned}
\end{equation*}
\end{proof}

\subsubsection{The proof of Theorem~\ref{theorem4}}
\begin{proof}
It is indifferent between varifier and free-rider for both validators, that is, 
\begin{equation}
\begin{aligned}
    & (1-\beta)\alpha_2\frac{1}{2}T + (1-\beta)(1-\alpha_2)\frac{V_{max}T}{V_{max}+{V_2}^*} + \beta(1-\alpha_2)\delta S + \alpha_2 \beta \frac{1}{2}\delta S -C  \\
    = & (1-\beta)\alpha_2\frac{TV_1^*}{V_1^*+V_{max}} + (1-\beta)(1-\alpha_2)\frac{V_1^*T}{V_1^*+{V_2}^*} + (1-\alpha_2)\beta\frac{{V_1^*}T}{{V_1^*}+{V_2}^*} - \alpha_2\beta f_pV_1^*, 
    \label{eqev1=ev1'}
\end{aligned}
\end{equation}

\begin{equation}
\begin{aligned}
    & (1-\beta)\alpha_1\frac{1}{2}T + (1-\beta)(1-\alpha_1)\frac{V_{max}T}{{V_1}^* + V_{max}} + \beta(1-\alpha_1)\delta S + \alpha_1 \beta \frac{1}{2}\delta S -C \\
    = & (1-\beta)\alpha_1\frac{TV_2^*}{V_2^*+V_{max}} + (1-\beta)(1-\alpha_1)\frac{V_2^*T}{V_2^*+{V_1}^*} + (1-\alpha_1)\beta\frac{{V_2^*}T}{{V_2^*}+{V_1}^*} - \alpha_1\beta f_pV_2^*
    \label{eqev2=ev2'}
\end{aligned}
\end{equation}

For the aggregator, it is indifferent between attacking and not attacking, that is, 
\begin{equation}
\begin{aligned}
    (1-\alpha_1)(1-\alpha_2)Z - (1-(1-\alpha_1)(1-\alpha_2))S = B
    \label{eqea=ea'}
\end{aligned}
\end{equation}

Assuming that validtor $P_1$ and $P_2$ are symmetric, Eq.(\ref{eqev1=ev1'}) and Eq.(\ref{eqev2=ev2'}) can be simplified into one equation, that is, 
\begin{equation}
\begin{aligned}
    & (1-\beta)\alpha\frac{1}{2}T + (1-\beta)(1-\alpha)\frac{V_{max}T}{V_{max}+{V}^*} + \beta(1-\alpha)\delta S + \alpha \beta \frac{1}{2}\delta S -C\\
    = &(1-\beta)\alpha\frac{TV}{V+V_{max}} + (1-\alpha)\frac{T}{2} - \alpha\beta f_pV.
    \label{eqev=ev'}
\end{aligned}
\end{equation}

    Combining Eq.(\ref{eqev=ev'}) and Eq.(\ref{eqea=ea'}), we get, 
    \begin{equation*}
    \begin{aligned}
        \alpha=1-\sqrt{\frac{B+S}{Z+S}},
    \end{aligned}
    \end{equation*}
    \begin{equation*}
    \begin{aligned}
        \beta=\frac{C-\frac{V_{max}T}{V_{max}+V^*}+\frac{1}{2}T}{ \delta S(1-\frac{1}{2}\alpha)+\alpha(\frac{1}{2}T+f_pV^*)-\frac{V_{max}T}{V_m+V^*}}.
    \end{aligned}
    \end{equation*}

The system loss is positively related to $\beta$, since, 
\begin{equation*}
\begin{aligned}
    \mathcal{L}=\beta\prod_{i=1}^n(1-\alpha_i)Z=\beta \frac{(S+B)Z}{S+Z}.
\end{aligned}
\end{equation*}

The aggregator's probability of attacking $\beta_{KYC}$ in the KYC scenario is, 
\begin{equation*}
\begin{aligned}
    \beta_{KYC} = \frac{C}{\delta S (1-\frac{1}{2} \alpha) + \alpha(\frac{1}{2}T+f_pV) - \frac{1}{2}T}.
\end{aligned}
\end{equation*}

And The aggregator's probability of attacking $\beta_{KYC}$ in the non-KYC scenario is,
\begin{equation*}
\begin{aligned}
    \beta_{nonKYC}=\frac{C-\frac{V_{max}T}{V_{max}+V^*}+\frac{1}{2}T}{ \delta S(1-\frac{1}{2}\alpha)+\alpha(\frac{1}{2}T+f_pV^*)-\frac{V_{max}T}{V_m+V^*}}.
\end{aligned}
\end{equation*}

Since $\beta_{KYC} > \beta_{nonKYC}$ when the fixed deposit $V$ in the KYC scenario is selected as $V^*$ in non-KYC scenario, $\mathcal{L}_{KYC} > \mathcal{L}_{non-KYC}$.
\end{proof}

Due to space limitations, Lemma~\ref{lemma4}, Lemma~\ref{lemma5}, and Lemma~\ref{lemma6}, which are purely mathematical proofs, are omitted here.
\end{document}